\long\def\onlyarxive#1{#1}
\documentclass[a4paper,12pt,twoside]{amsart}

\usepackage{amssymb,epsfig}
\usepackage{hyperref}
\usepackage{amsmath}
\usepackage{amscd}
\usepackage{amsthm, color}
\usepackage{rotating}

\usepackage{enumerate}

\newtheorem{thm}{Theorem}[section]
\newtheorem{lem}[thm]{Lemma}
\newtheorem{lemma}[thm]{Lemma}
\newtheorem{prop}[thm]{Proposition}

\newtheorem{cor}[thm]{Corollary}

\theoremstyle{definition}
\newtheorem{defn}[thm]{Definition}
\newtheorem{definition}[thm]{Definition}

\newtheorem{remark}[thm]{Remark}

\newcommand{\End}{\operatorname{End}}

\newcommand{\Diff}{\operatorname{Diff}}
\newcommand{\Diffeo}{\operatorname{Diffeo}}
\newcommand{\CI}{C^\infty}
\newcommand{\CIc}{C^\infty_c}
\newcommand{\grad}{\operatorname{grad}}
\newcommand{\id}{\operatorname{id}}
\def\hookdownarrow{\mbox{\begin{turn}{270}$\hookrightarrow\kern4mm$\end{turn}}}

\newcommand{\pa}{\partial}
\newcommand\vect[1]{\overset{\to}{ \mathbf{#1}}}

\let\witi\widetilde

\renewcommand{\tilde}{\widetilde}
\renewcommand{\bar}{\overline}

\newcommand{\NN}{\mathbb{N}}
\newcommand{\ZZ}{\mathbb{Z}}

\newcommand{\RR}{\mathbb{R}}
\newcommand{\CC}{\mathbb{C}}
\renewcommand{\SS}{\mathbb{S}}

\def\maB{\mathcal{B}}

\def\maF{\mathcal{F}}

\def\maK{\mathcal{K}}

\def\maS{\mathcal{S}}  
\def\maX{\mathcal{X}}  
\def\maY{\mathcal{Y}}            
\def\maV{\mathcal{V}} 
\def\maW{\mathcal{W}} 
\def\maT{\mathcal{T}}

\let\cB\maB
\let\cH\maH 

\newcommand{\supcirc}{\supset\hspace{-0.4cm}\circ}

\newcommand\ie{{i.\kern2pt e.\ }}
\newcommand\vol{\operatorname{vol}}

\title{Regularity for eigenfunctions of Schr\"{o}dinger operators}

\begin{document}

\author[B. Ammann]{Bernd Ammann}
       \address{B. Ammann, Fakult\"at f\"ur Mathematik, Universit\"at
       Regensburg, 93040 Regensburg,
       Germany} \email{bernd.ammann@mathematik.uni-regensburg.de}
       
\author[C. Carvalho]{Catarina Carvalho}
       \address{C. Carvalho, Instituto Superior T\'ecnico,
       Math. Dept., UTL, Av. Rovisco Pais, 1049-001 Lisbon,
       Portugal} \email{ccarv@math.ist.utl.pt}

\author[V. Nistor]{Victor Nistor} \address{V. Nistor, Pennsylvania State
       University, Math. Dept., University Park, PA 16802,
       USA} \email{nistor@math.psu.edu}

\thanks{Ammann's manuscripts are available from {\bf
       http:{\scriptsize//}www.berndammann.de{\scriptsize/}publications}. Carvalho's
       manuscripts are available from {\bf
       http:{\scriptsize//}www.math.ist.utl.pt/$\sim$ccarv}.  Nistor's Manuscripts available from {\bf
       http:{\scriptsize//}www.math.psu.edu{\scriptsize/}nistor{\scriptsize/}}.
       Nistor was partially supported by the NSF Grants DMS-0713743, 
       OCI-0749202, and DMS-1016556. }



\begin{abstract} 
We prove a regularity result in weighted Sobolev (or
Babu\v{s}ka--Kondratiev) spaces for the eigenfunctions of
certain Schr\"{o}dinger-type operators. Our results apply, in
particular, to a non-relativistic Schr\"{o}dinger operator of an
$N$-electron atom in the fixed nucleus approximation.  More
precisely, let $\maK_{a}^{m}(\RR^{3N},r_S)$ be the weighted Sobolev space
obtained by blowing up the set of singular points of the potential
$V(x) = \sum_{1 \le j \le N} \frac{b_j}{|x_j|} + \sum_{1 \le i < j \le
N} \frac{c_{ij}}{|x_i-x_j|}$, $x \in \RR^{3N}$, $b_j,
c_{ij} \in \RR$. If $u \in L^2(\RR^{3N})$ satisfies $(-\Delta + V) u
= \lambda u$ in distribution sense, then $u \in \maK_{a}^{m}$ for all
s$m \in \ZZ_+$ and all $a \le 0$. Our result extends to the case when
$b_j$ and $c_{ij}$ are suitable bounded functions on the blown-up
space. In the single-electron, multi-nuclei case, we obtain the same
result for all $a<3/2$.
\end{abstract}

\maketitle
\tableofcontents


\section{Introduction}

We prove a global regularity result for the eigenfunctions
of a non-relativistic Schr\"{o}dinger operator $\cH := -\Delta +
V$ of an $N$-electron atom.  More
precisely, let
\begin{equation} \label{eq.def.V}
  V(x) = \sum_{1 \le j \le N} \frac{b_j}{|x_j|} 
  + \sum_{1 \le i < j \le N} \frac{c_{ij}}{|x_i-x_j|},
\end{equation}
where $x = (x_1, x_2, \ldots, x_N) \in \RR^{3N}$,
$x_j \in \RR^3$, and $b_j$ and $c_{ij}$ are suitable smooth
functions. This potential can be used to model the case of a single,
heavy nucleus, in which case $b_j$ are negative constants, arising
from the attractive force between the nucleus and the $j$-th electron,
whereas the $c_{ij}$ are positive constants, arising from the
repelling forces between electrons. Our results, however, will not
make use of sign assumptions on the coefficients $b_j$,
$c_{ij}$. We also study the case of one electron
and several fixed nuclei, which is important for the study of Density
Functional Theory, Hartree, and Hartee-Fock equations. In that case,
our regularity results are optimal. Our method can also be applied to the
case of several light nuclei and to the study of wave packets, as in \cite{Kato}.

Let $u \in L^2(\RR^{3N})$ be an eigenfunction of $\cH := -\Delta + V=
- \sum_{i=1}^{3N} \frac{\partial^2}{\partial x_i^2} +V$, the
Schr\"odinger operator associated to this potential, that is, a
non-trivial solution of
\begin{equation}\label{hamilt}
  \cH u := -\Delta u + V u = \lambda u
\end{equation} 
in the sense of distributions, where $\lambda \in \RR$. Our main goal
is to study the regularity of $u$. One can replace the Laplacian
$\Delta$ with any another uniformly strongly elliptic operator on
$\RR^n$.  Typically the negativity of the $b_j$ implies that
infinitely many eigenfunctions of~$\cH$ exist, see for instance the
discussion in \cite[XIII.3]{Simon4}. In physics, an eigenfunction of~$\cH$ 
(associated to a discrete eigenvalue with finite multiplicity) is
interpreted as a bound electron, as its evolution under the
time-dependent Schr\"odinger equation is $e^{-i\lambda t}u(x)$ and
thus the associated probability distribution $|u(x)|^2$ does not
depend on~$t$.


The potential $V$ is singular on the set $S := \bigcup_j \{ x_j =
0 \} \cup\bigcup_{i < j} \{ x_i = x_j \}$. The planes in the union
defining $S$ describe the collision of at least two particles, thus we
also call them collision planes, as customary.  Basic elliptic
regularity \cite{EvansBook, taylor2} shows that $u \in H^s_{\rm
loc}(\RR^{3N} \smallsetminus S)$ for all $s \in \RR$, which is however
not strong enough for the purpose of approximating the eigenvalues and
eigenvectors of $\cH$.  Moreover, it is known classically that $u$ is
not in $H^s(\RR^{3N})$ for all $s \in
\RR$ \cite{Fournais2,HunsickerNistorSofo,taylor2}. 
If the coefficient $b_j$ and $c_{ij}$ are real-analytic, then it
follows from analytic regularity theory (see
e.g.\ \cite[Theorem~6.8.1]{morrey:66}) that $u$ is analytic on
$\RR^{3N}\setminus S$. In this case a strong local regularity result
was obtained in \cite{Fournais2} in the neighborhood of the simple
coalescence points, where it was shown that locally $u(x) = u_1(x) +
|l(x)| u_2(x)$ with $u_1$ and $u_2$ real analytic and $l$ linear.  See
also \cite{CarmonaReg, CarmonaSimon2, jecko2, cfks, Flad2, Flad1,
Flad3, Fournais3, Fournais1, Jecko, Kato, VasyReg,
Yserentant1, Yserentant2, Yserentant3} and references therein for more
results on the regularity of the eigenfunctions of Schr\"{o}dinger
operators. Related is~\cite{flad.harutyunyan} which was circulated
after this article has been submitted.

Our approach is to use the ``weighted Sobolev spaces,'' or
``Babu\v{s}ka-Kondratiev spaces,''
\begin{equation}\label{eq.def.sobw}
        \maK_{a}^m(\RR^n,r_S) := \{ u: \RR^{n} \to \CC \,|\,
        r_S^{|\alpha| - a} \pa^\alpha u \in L^2(\RR^n), \ |\alpha| \le
        m \},
\end{equation}
where $a\in\RR$ and $m\in\NN\cup\{0\}$. The weight function $r_S(x)$
is the smoothed distance from $x$ to $S$, however the distance $r_S$
is not measured with respect to euclidean distance, but with respect
to a metric on the ball compactification of $\RR^n$. This modified
choice does not effect $\maK_{a}^m(\RR^n, r_S)$ on closed balls, but
globally.  The main result of our paper (Theorem \ref{theorem.main})
is that
\begin{equation}
\label{eq.main}
        u \in \maK_{a}^m(\RR^{3N},r_S)
\end{equation}
for $a \le 0$ and for {\em arbitrary} $m\in \NN$.  For a single
electron, we prove the same result for $a<3/2$ and conjecture that
this holds true in general. Let us notice that we obtain higher regularity
results, which were not available before (for instance, the results in \cite{Kato},
yield the boundedness of eigenfunctions and of their gradients, but no results
on the higher derivatives).

The proof of our main result uses a suitable compactification $\SS$ of
$\RR^{3N} \smallsetminus S$ to a manifold with corners, which turns
out to have a Lie manifold structure. Then we use the regularity
result for Lie manifolds proved in \cite{ain}. The weighted Sobolev
spaces $\maK_{a}^m(\RR^{3N},r_S)$ then identify with some
geometrically defined Sobolev spaces (also with weight).

To obtain the space $\SS$, we compactify $\RR^{3N}$ to a ball.  This
ball carries a Lie manifold structure, which describes the geometry
underlying the scattering calculus, recalled later.  The space $\SS$
is then blown up along the closure of the singular set~$S$ in
$\RR^{3N}$.  For this we decompose $S$ in its strata of different
dimensions and then blow up the $0$-dimensional stratum first and then
successively the strata of higher and higher dimension.  The resulting
compact space is a manifold with corners~$\SS$ whose interior is
naturally diffeomorphic to $\RR^{3N} \setminus S$. Roughly speaking,
the blow-up-compactification procedure amounts to define generalized
polar coordinates close to the singular set in which the analysis
simplifies considerably.  Each singular stratum of the singular set
$S$ gives rise to a boundary hyperface, corresponding to the collision planes, 
in the blown-up
manifold with corners~$\SS$, and the distance functions to the strata
turn into boundary defining functions. (These kind of hyperfaces are called also 'hyperfaces at inifinity'.)
The construction of the manifold with corners $\SS$ 
is a standard technique, see e.g.\
\cite{melrose.dam} for a similar construction.

We show then that, additionally, the compactification $\SS$ carries a Lie
structure at infinity~$\maW$, a geometric structure developed in
\cite{aln1, ammann.lauter.nistor:07}, which extends work by Melrose, 
Schrohe, Schulze, Vasy and their collaborators, which in turn build on
earlier results by Cordes~\cite{Cordes}, Parenti~\cite{Parenti}, and
others.  More precisely, $ \maW$ is a Lie subalgebra of vector fields
on $\SS$ with suitable properties (all vector fields are tangent to
the boundary,  $ \maW$  is a finitely generated projective
$\CI(\SS)$-module, there are no restrictions on~$\maW$ in the
interior of $\SS$). There is a natural algebra $\Diff_{\maW}(\SS)$ of
differential operators on $\SS$, defined as the set of differential
operators generated by $\maW$ and $\CI(\SS)$. This Lie structure
is obtained interatively as well. On the ball compactification of $\RR^{3N}$
this is just the Lie structure underlying the scattering calculus. 
We will show in Section~\ref{section.lie.struct} that each time we blow up 
a Lie manifold along a suitable submanifold, then the blown-up manifold
inherits the structure of a Lie manifold as well. In particular,
we obtain a Lie manifold structure on the blown-up manifold without assuming 
any additional condition at infinity, in contrast to the existing
literature where Lie manifold structures on 
blow-ups have only been developed under quite restrictive conditions.
The Lie structure on~$\SS$ provides a Lie algebroid $A$ on $\SS$, 
a structure which, in particular, is a vector bundle $A$ over $\SS$.
It comes with an anchor map $\rho:A\to T\SS$, a vector bundle
homomorphism which is an isomorphism in the interior of $\SS$.  A
metric on $A$ gives rise to a complete metric~$g$ on the
interior~$\SS_0$ of~$\SS$. Metrics on~$\SS_0\cong\RR^{3N}\setminus S$
obtained this way are said to be compatible with the Lie structure.
Our blow-up procedure yields such a compatible metric on
$\RR^{3N}\setminus S$ with the additional property to be conformal to
the euclidean metric.

Our analytical results will be obtained by studying the properties of
the differential operators in $\Diff_{\maW}(\SS)$ and then by relating
our Hamiltonian to $\Diff_{\maW}(\SS)$. Here $\maW$ is the Lie
algebra of vector fields defining the Lie manifold structure of
$\SS$. Some of the relevant results in this setting were obtained
in~\cite{ain}. More precisely, let $\rho := \prod_{1\leq i\leq k}
x_{H_i}$, where $\cB=\{H_1,\ldots,H_k\}$ is the set of boundary
hyperfaces of $\SS$ that
are obtained by blowing up the singular set (corresponding to the
collision planes) and $x_{H_i}$ is a defining function of the
hyperface $H_i$. An important step in our approach is to show that
$\rho^2\cH \in \Diff_{\maW}(\SS)$, where $\cH=-\Delta +V$ is as in
Equation~\eqref{hamilt} (see Theorem \ref{thm.regularity}).

Let $H^m(\SS)$ be the Sobolev spaces associated to a metric $g$ on
$\RR^{3N} \smallsetminus S$ compatible with the Lie manifold structure
on $\SS$, namely
\begin{equation}\label{eq.def.Hm}
        H^m(\SS):=\{u\in L^2(\RR^{3N}) \,|\, Du\in
        L^2(\RR^{3N}\smallsetminus S, d\vol_g ),\ \forall\ 
        D\in \Diff^m_{\maW}(\SS)\}.
\end{equation}
For any vector $\vect{a}=(a_H)_{H\in \cB} \in \RR^k$, where again
$k:=\# \cB$ is the number of hyperfaces of $\SS$ corresponding to the
singular set (the collision planes), we
define $H_{\vect{a}}^m(\RR^{3N}) := \chi H^m(\SS)$, with $\chi
:= \prod_{H\in \cB} x_H^{a_H}$. In particular,
$H_{\vect{0}}^m(\RR^{3N}) = H^m(\SS)$.  This allows us to use the
regularity result of \cite{ain} to conclude that $u \in
H_{\vect{a}}^m(\RR^N)$ for all $m$, whenever $u \in
H_{\vect{a}}^0(\RR^N)$. Since $H_{\vect{a}}^0(\RR^{3N}) =
L^2(\RR^{3N})$ for suitable $\vect{a} = (a_H)$, this already leads to
a regularity result on the eigenfunctions $u$ of $\cH$, which is
however not optimal in the range of $a$, as we show for the case of a
single electron (but multiple nuclei). Future work will therefore be
needed to make our results fully applicable to numerical methods. One
will probably have to consider also regularity in anisotropically
weighted Sobolev spaces as in~\cite{bnz3}.

We now briefly review the contents of this paper. In
Section~\ref{sec.diff.struc}, we describe the differential structure
of the blow-up of a manifold with corners by a family of submanifolds
that intersect cleanly. In particular, we define the notion of
iterated blow-up in this setting. In Section~\ref{section.lie.struct},
we review the main definitions of manifolds with a Lie structure at
infinity and of lifting vector fields to the blown-up manifold. The
main goal is to show that the iterated blow-up of a Lie manifold
inherits such a structure (see Theorem~\ref{theorem.Lie.P}). We give
explicit descriptions of the relevant Lie algebras of vector fields,
study the geometric differential operators on blown-up spaces and
describe the associated Sobolev spaces. Finally, in
Section~\ref{sec.reg.eigen}, we consider the Schr\"odinger operator
with interaction potential~\eqref{eq.def.V} and apply the results of
the previous sections to obtain our main regularity result,
Theorem \ref{theorem.main}, whose main conclusion is
Equation~\eqref{eq.main} stated earlier. The range of the index $a$ in
Equation~\eqref{eq.main} is not optimal. New ideas are needed to
improve the range of $a$. We show how this can be done for the case of
a single electron, but multi-nuclei, in which case we do obtain the
optimal range $a<3/2$. When $b_j$ and $c_{ij}$ are constants, our
regularity result in the single electron case is also a consequence of
\cite{Fournais3,Fournais2}.

In fact, for the case of a single electron and several nuclei, our
result is more general, allowing for the potentials that arise in
applications to the Hartree-Fock equations and the Density Functional
Theory. As such, they can be directly used in applications to obtain
numerical methods with optimal rates of convergence in $\RR^3$. For
several electrons, even after obtaining an optimal range for the
constant $a$, our results will probably need to be extended before
being used for numerical methods. The reason is that the resulting
Riemannian spaces have exponential volume growth. This problem can be
fixed by considering anisotropically weighted Sobolev spaces, as
in \cite{bnz3}. The results for anisotropically weighted Sobolev
spaces however are usually a consequence of the results for the usual
weighted Sobolev spaces. For several electrons, one faces additional
difficulties related to the high dimension of the corresponding space
(curse of dimensionality).

\subsection{Acknowledgements} We thank Eugenie Hunsicker, Jorge Sofo 
and Daniel Grieser for useful discussions. B. Ammann and V. Nistor
thank Werner Ballmann and the Max-Planck-Institut for Mathematics in
Bonn, Germany, for its hospitality. We are also grateful to an unknown
referee for his helpful comments.

\section{Differential structure of blow-ups}\label{sec.diff.struc}

\subsection{Overview}
The main goal of this section is to establish a natural procedure to
desingularize a manifold with corners $M$ along finitely many
submanifolds $X_1,X_2,\ldots,X_k$ of $M$.  This construction is often
useful in studying singular spaces such as polyhedral domains and
operators with singular potentials \cite{ammann.nistor:07, bmnz,
Hunsicker, KMR}. Its origins in the setting of pseudodifferential
calculus on singular spaces can be traced to the work of Melrose,
Schulze, and their collaborators, building on earlier work by Cordes,
Parenti, Taylor, and others, see~\cite{MelrosePCSL,
MazzeoMelrose,grieser:b} and references therein.  See also the
notes \cite{melrose.dam} for more on the constructions below.  In what
follows, by a {\em manifold} we will mean a manifold that may have
corners. On the other hand, by a {\em smooth manifold} we shall
understand a manifold that does {\bf not} have a boundary (so no
corners either). In addition, a submanifold is always required to be
a \emph{closed} subset.

If $X$ is a submanifold of $M$, then the desingularization procedure
yields a new manifold, called the \emph{blow-up} of $M$ along $X$,
denoted by $[M:X]$. Roughly speaking, $[M:X]$ is obtained by removing
$X$ from $M$ and gluing back the unit sphere bundle of the normal
bundle of $X$ in $M$. If $M$ is a manifold without boundary, then
$[M:X]$ is a manifold whose boundary is the total space of that sphere
bundle. More details will be given below. There is also an associated 
natural {\em blow-down map} $\beta:[M:X]\to M$ which is the identity
on $M\setminus X$.

Then we want to desingularize along a second submanifold $X'$ of $M$,
typically we will have $X\subset X'\subset M$.  In this situation, the
inclusion $X'\hookrightarrow M$ lifts to an embedding
$[X':X]\hookrightarrow [M:X]$. Then we blow-up $[M:X]$ along $[X':X]$,
obtaining a manifold with corners. 
%
%
An iteration will then yield the desired blown-up manifold. Since we
are interested in applying our results to the Schr\"odinger equation,
we have to allow that submanifolds intersect each other. These
intersection will be blown up first before the submanifold themselves
are blown up. So even if one is interested just in smooth manifolds
without boundary, a repeated blow up will lead to manifolds with
corners.

\subsection{Blow-up in smooth manifolds}

It is convenient to first understand some simple model cases.  If
$M=\RR^{n+k}$ and $X=\RR^n\times \{0\}$, then we define
\begin{equation}\label{eq.blup.Rn}
        [\RR^{n+k}:\RR^n\times \{0\}]:=\RR^n \times S^{k-1} \times [0,\infty),
\end{equation}
with blow-down map 
\begin{equation}
\label{eq.local}
       \beta: \RR^n \times S^{k-1} \times
         [0,\infty) \to \RR^{n+k}, \quad (y,z,r) \mapsto (y,zr).
\end{equation} 
If $x\in \RR^n \times S^{k-1} \times (0,\infty)$, then we identify $x$
with $\beta(x)$, in the sense that $\RR^n \times S^{k-1} \times
(0,\infty)$ is interpreted as polar coordinates for
$\RR^{n+k}\setminus \RR^n$. In the following we use the symbol
$\sqcup$ for the {\em disjoint} union. We obtain (as sets)
\begin{equation*}
        [\RR^{n+k}:\RR^n\times \{0\}]
        = (\RR^{n+k}\setminus \RR^n \times\{0\})\sqcup \RR^n\times
        S^{k-1}.
\end{equation*}

\begin{remark}
An alternative way to define $[\RR^{n+k}:\RR^n\times \{0\}]$ is as
follows. For any $v\in\RR^{n+k}\setminus \RR^n\times \{0\}$ define the
$(n+1)$-dimensional half-space $E_v:=\{x+tv\,|\,
x\in \RR^n\times \{0\},\ t\geq 0\}$ and
${G:=\left\{E_v\,|\,v\in\RR^{n+k}\setminus \RR^n\times \{0\}\right\}\cong
S^{k-1}}$.  Then
\begin{equation*}  
                   [\RR^{n+k}:\RR^n\times \{0\}] := \{(x,E)\,|\, E\in
                   G,\; x\in E\}
\end{equation*} 
and $\beta(x,E):=x$.  The equation $x\in E$ defines a submanifold with
boundary of $\RR^{n+k}\times G$, and its boundary is $\{(x,E)\,|\,
E\in G, \;x\in \RR^{n} \times \{0\}\}\cong \RR^n\times S^{k-1}$.
\end{remark}

If $V$ is an open subset of $\RR^{n+k}$ and $X=(\RR^n\times \{0\})\cap V$,
then  the blow-up of $V$ along $X$ is defined as
\begin{equation*}  
                   [V:X]:=\beta^{-1}(V) =V\setminus X \sqcup \beta^{-1}(X)
\end{equation*}
for the above map $\beta:[\RR^{n+k}:\RR^n\times \{0\}]\to \RR^{n+k}$,
and the new blow-down map is just the restriction of $\beta$ to
$[V:X]$.
  
\begin{lemma} \label{lemma.loc}
Let $\phi:V_1\to V_2$ be a diffeomorphism between
two open subsets of $\RR^{n+k}$, mapping $X_1:= V_1\cap\RR^n\times \{0\}$ onto 
$X_2:= V_2\cap\RR^n\times \{0\}$. 
Then $\phi$ uniquely lifts to a
diffeomorphism
\begin{equation*}
  \phi^\beta : [V_1:X_1] \to [V_2:X_2]
\end{equation*}
covering $\phi$ in the sense that $\beta \circ \phi^\beta = \phi \circ
\beta$.
\end{lemma}


\begin{proof} 
For $x\in V_1\setminus X_1\subset [V_1:X_1]$ we set
$\phi^\beta(x):=\phi(x)$.  
Elements in $\beta^{-1}(X_1)$ will be written 
as $(x,v)$ 
with $x=\beta(x,v)\in X_1\subset \RR^n$ and 
$v\in S^{k-1}\subset \RR^k$.  Note that
$d_{x}\phi\in\End(\RR^{n+k})$ maps $\RR^n\times\{0\}$ to
itself, and thus has block-form
\begin{equation*}
        \begin{pmatrix}A& B\\ 0 &D\end{pmatrix}.
\end{equation*}
We then define $\phi^\beta(x,v):=(\phi(x),\frac{D
v}{\|Dv\|},0)\in \RR^{n}\times S^{k-1}\times [0,\infty)$.  The
smoothness of $\phi^\beta: [V_1:X_1] \to [V_2:X_2]$ can be checked in
polar coordinates.  

Alternatively using the above remark, one can
express this map as $\phi^\beta(x,E_x)=(\phi(x), E_{\phi(x)})$ for
$x\in V_1\setminus X_1$ and $\phi^\beta(x,E):=(\phi(x),d_x\phi(E))$ if
$x\in X_1$.  In this alternative expression the smoothness of
$\phi^\beta$ is an immediate consequence of the definition of
derivative as a limit of difference quotients.
\end{proof}

Now let $M$ be an arbitrary smooth manifold (without boundary) of
dimension $n+k$ and $X$ a (closed) submanifold of $M$ of dimension
$n$. We choose an atlas $\mathcal{A}:=\{\psi_i\}_{i\in I}$ of $M$
consisting of charts $\psi_i:U_i\to V_i$ such that $X_i:=X\cap
U_i=\psi_i^{-1}\big (V_i \cap (\RR^n\times\{0\}) \big)$. Note that we
do not exclude the case $X\cap U_i=\emptyset$.  Then the previous
lemma tells us that the transition functions
\begin{equation*}
        \phi_{ij} := \psi_i\circ \psi_j^{-1} : V_{ij}
        := \psi_j(U_i\cap U_j)\to V_{ji}:=\psi_i(U_i\cap U_j)
\end{equation*} 
can be lifted to maps
\begin{equation*}
        \phi_{ij}^\beta: [V_{ij}: X_{ij}]\to [V_{ji}: X_{ji}],
\end{equation*} 
where $X_{ij}:= \psi_j(U_i\cap U_j \cap X)$.

Gluing the manifolds with boundary $[V_i: X_i]$, $i\in I$ with respect
to the maps $\phi_{ij}^\beta$, $i,j\in I$ we obtain a manifold with
boundary denoted by $[M:X]$ and gluing together the blow-down maps
yields a map $\beta:[M:X]\to M$. The boundary of $[M:X]$ is
$\beta^{-1}(X)$.  The restriction of $\beta$ to the interior
$[M:X]\setminus \beta^{-1}(X)$ is a diffeomorphism onto $M\setminus X$
which will be used to identify these sets.

Recall that the {\em normal bundle} of $X$ in $M$ is the bundle
$N^MX\to X$, whose fiber over $p\in X$ is the quotient
$N_p^MX:=T_pM/T_pX$. Fixing a Riemannian metric $g$ on $M$, the normal
bundle is isomorphic to $T^\perp X=\{v\in T_pM\,|\,p\in X,\quad v\perp
T_pX\}$. We shall need also the normal sphere bundle $S^MX$ of $X$ in
$M$, that is, the sphere bundle over $X$ whose fiber $S_p^MX$ over
$p\in X$ consists of all unit length vectors in $N_p^MX$ with respect
the metric on $N^MX$. The choice of $g$ will not affect our
construction. The restriction of $\beta|_{\beta^{-1}(X)}
: \beta^{-1}(X)\to X$ is a fiber bundle over $X$ with fibers
$S^{k-1}$, which is isomorphic to the normal sphere bundle.

Let us summarize what we know about the blow-up $[M:X]$ thus obtained.
As sets we have $[M:X]=M\setminus X \sqcup S^MX$. The set $S^MX$ is
the boundary of $[M:X]$, and the exact way how this boundary is
attached to $M\setminus X$ is expressed by the lifted transition
functions $\phi_{ij}^\beta$.  More importantly, we have seen that the
construction of the blow-up is a local problem, a fact that will turn
out to be useful below when we discuss the blow-up of manifolds with
corners.

\subsection{Blow-up in manifolds with corners}

Now let $M$ be an $m$-dimensional manifold with corners. Recall that
by a hyperface of $M$ we shall mean a boundary face of
codimension~$1$.  The intersection of $s$ hyperfaces
$H_1\cap\ldots \cap H_s$, if non-empty, is then a union of boundary
faces of codimension $s$ of $M$. We shall follow the definitions and
conventions from \cite{aln1}. In particular, we shall always assume
that each hyperface is embedded and has a defining function. We also
say that points $x$ in the interior of $H_1\cap\ldots \cap H_s$ are
points of boundary depth $s$, in other word the boundary faces of
codimension $k$ contain all points of boundary depth $\geq k$.  Points
in the interior of $M$ are points of boundary depth $0$ in $M$. In the
case $s=0$ the intersection $H_1\cap \ldots\cap H_s$ denotes $M$.

\begin{defn}\label{def.subm}
A closed subset $X\subset M$ is called a
\textit{submanifold with corners} of codimension $k$ if
any point $\bar x\in X$ of boundary depth $s\in \NN\cup\{0\}$ in
$M$ has an open neighborhood $U$ in $M$ and smooth functions
$y_1,\ldots,y_k:U\to \RR$ such that the following hold:
\begin{enumerate}[{\rm (i)}]
\item $X\cap U=\{x\in U\,|\, y_1(x)=y_2(x)=\cdots=y_k(x)=0\}$
\item 
Let $H_1,\ldots, H_s$ be the boundary faces containing $\bar x$ (which
is equivalent to saying that $\bar x$ is in the interior of $X\cap
H_1\cap \ldots \cap H_s$). Let $x_1,\ldots, x_s$ be boundary
defining functions of $H_1,\ldots,H_s$.  Then $dy_1, \ldots, dy_k,
dx_1, \ldots, dx_s$ are linearly independent at $\bar x$.
\end{enumerate}
\end{defn}

\begin{remark} Similar notions were also introduced and studied by Melrose
in  \cite{melrose.dam},
however with a different aim and a slightly different terminology. 
A submanifold
with corners in the above sense, is the same as a p-submanifold with $l=k$
in  \cite[Sec. 1.7]{melrose.dam}, and this is equivalent to an interior
p-submanifold in later sections of~\cite{melrose.dam}. Such blow-ups 
are iterated in \cite{melrose.dam} as well, and the iterated constructions
coincide with our iterated blow-up described below 
in the case of chains. However, in contrast to \cite{melrose.dam},
if a clean family of submanifolds 
(definition see below) contains submanifolds $X_1$ and $X_2$ with
$X_1\not\subset X_2$ and $X_2\not\subset X_1$, we will always blow-up $X_1\cap
X_2$ before blowing up $X_1$ and $X_2$ which yields stronger 
analytic properties.
\end{remark}

A simple example of a submanifold with corners $X$ of a manifold with
corners $M$ is
\begin{equation*}
        X:=[0,\infty)^{m-k}\times \{0\} \subset
        M:=[0,\infty)^{m-k}\times \RR^k.
\end{equation*}
Here the codimension is $k$, and as $y_i$ we can choose the standard
coordinate functions of $\RR^k$, and as $x_i$ the coordinate functions
of $[0,\infty)^{m-k}$.

On the other hand this simple example already provides models for all
kind of local boundary behavior of a submanifold with corners $X$ of a
manifold with corners $M$ with codimension $k$, and $m=\dim M$.  More
precisely, a subset $X$ of a manifold with corners $M$ is a
submanifold with corners in the above sense if, and only if, any $x\in
X$ has an open neighborhood $U$ and a diffeomorphism $\phi:U\to V$ to
an open subset $V$ of $[0,\infty)^{m-k}\times \RR^k$ with $\phi(X\cap
U)= ([0,\infty)^{m-k}\times \{0\})\cap V$.

As before, all submanifolds with corners shall be {\em closed} subsets
of $M$, contrary to the standard definition of a smooth submanifold of
a smooth manifold.  The definition of a submanifold with corners gives
right away:
\begin{enumerate}[{\rm (i)}]
\item \emph{Interior submanifold}: the interior of $X$ is a closed submanifold 
of codimension $k$ of the interior of $M$, in the usual sense.
\item \emph{Constant codimension}: If $F$ is the interior of a boundary
  face of $M$ of codimension~$s$, then $F\cap X$ is an
  $(m-k-s)$-dimensional submanifold of $F$, that is, $F \cap X$ is
  also a submanifold (in the usual sense)
of codimension $k$ in $F$.
\item \emph{Clean intersections}: If $F$ is as above and $x\in F\cap X$, 
then $T_x(F\cap X)=T_x F\cap T_x X$
\end{enumerate} 
The use of the term ``clean'' goes back to the work of Bott, and was then
used again in \cite{melrose.dam}.

Let $N^MX$ denote the normal bundle of $X$ in $M$. Now, if $F$ is the
interior of a boundary face, then the inclusion $F\hookrightarrow M$
induces a vector bundle isomorphism
\begin{equation*}
        N^F(X\cap F)\cong N^MX|_{X\cap F}.
\end{equation*}
Similarly, we obtain for the interior $F$ of any boundary face an
isomorphism for normal sphere bundles
\begin{equation*}
        S^F(X\cap F)\cong S^MX|_{X\cap F}.
\end{equation*}

Now we will see how to blow-up a manifold $M$ with corners along a
submanifold $X$ with corners. For simplicity of presentation let
$k\geq 1$.  As before, we have as sets $[M:X]=M\setminus X \sqcup
S^MX$, but here $M\setminus X$ will, in general, have boundary
components, each boundary face $F$ of $M$ will give rise to one (or
several) boundary faces for $[M:X]$. The total space of $S^MX$ yields
new boundary hyperfaces.

To construct the manifold structure on $[M:X]$ one can proceed as in
the smooth setting. Let $\beta : [\RR^{n+k} : \RR^{n} \times \{0\}]$
be the blow-down map. Then the blow-up of
\begin{equation*}
        \RR^{n-s}\times
        [0,\infty)^s\times \{0\} \subset \RR^{n-s}\times
        [0,\infty)^s\times \RR^k
\end{equation*} 
is just the restriction of $[\RR^{n+k}:\RR^n\times\{0\}]\to \RR^{n+k}$
to $\beta^{-1}(\RR^{n-s}\times [0,\infty)^s\times \RR^k)$.  Similarly,
Lemma~\ref{lemma.loc} still holds if $V_i$ are open subsets of
$\RR^{n-s}\times [0,\infty)^s\times \RR^k$, and gluing together charts
with the lifted transition functions $\phi_{ij}^\beta$ yields a
manifold with corners $[M:X]$ in a completely analogous way as in the
previous section.  In this way, we have defined $[M:X]$ if $M$ is a
manifold with corners, and if $X$ is a submanifold with corners of $M$.

For the convenience of the reader, we now describe an alternative way
to define $[M:X]$.  Let $\cB=\{H_1,\ldots,H_k\}$ be the set of
(boundary) hyperfaces of $M$. We first realize $M$ as the set
$\{x\in \tilde M\,|\, x_H \ge 0,\; \forall H\in\cB\}$, for $\tilde M$
an enlargement of $M$ to a smooth manifold, such that $X = \tilde
X \cap M$, for a smooth submanifold $\tilde X$ of $\tilde M$. Here
$\{x_H\}$ is a set of boundary defining functions of $M$, 
extended smoothly to $\tilde M$. Let $\beta
: [\tilde M: \tilde X ] \to \tilde M$ be the blow-down map. Then we
can define $[M: X] := \beta^{-1} (M)=\{x\in [\tilde M: \tilde X],
x_H(\beta(x)) \ge 0\}$, and, slightly abusing notation, we will write
again $x_H$ for $x_H\circ \beta$.  The definition of a submanifold
with corners ensures that $[M: X]$ is still a manifold with
corners. Note that smooth functions on $M$ (respectively $[M: X]$) are
given by restriction of smooth functions on $\tilde M$ (respectively
$[\tilde M: \tilde X]$).

It also is helpful to describe the set of boundary hyperfaces of
$[M:X]$.  Some of them arise from boundary hypersurfaces of $M$ and
some of them are new.  Let $H$ be a connected boundary hyperface of
$M$.  All connected components of $H\setminus (X\cap H)$ give rise to
a connected hyperface of $[M:X]$. The other connected hyperfaces of
$[M:X]$ arise from connected components of $X$. Each connected
component of $X$ yields a boundary hyperface for $[M:X]$, which is
diffeomorphic to the normal sphere bundle of $X$ restricted to that
component.  (Such a hyperface arising from $X$ is said to be an hyperface at infinity.)
The boundary hyperfaces of $X$ then induce codimension $2$
boundary faces for $[M:X]$ each of which is the common boundary of a
hyperface arising from $M$ and a hyperface arising from~$X$.

One can describe similarly the codimension~$2$ boundary faces of
$[M:X]$. Some of them arise from boundary hyperfaces of $X$, as
described in the paragraph above, the others arise from boundary faces
of $M$ of codimension~$2$. More precisely, let $F$ be the interior of
such a face, then any connected component of $F\setminus X$ is a
connected component of a boundary face of codimension~$2$ of $[M:X]$.

As for boundary defining functions, let $\bar g$ be a {\em true }
Riemannian metric on $M$, that is a smooth metric on $M$, defined and
smooth up to the boundary.  We shall denote by $r_X : M \to [0,
1]$ a continuous function on $M$, smooth outside $X$ that close to
$X$ is equal to the distance function to $X$ with respect to $\bar g$
and $r_X^{-1}(0) = X$. A function with these properties will be called
{\em a smoothed distance function to $X$}.  If $X$ and all $H\setminus
(X\cap H)$ are connected, then $x_H$, $H\in \cB$ and $r_X$ (identified
with their lifts to the blow-up) are boundary defining functions of
$[M:X]$. This statement generalizes in an obvious way to the
non-connected case.

\subsection{Blow-up in submanifolds}
For the iterated blow-up construction we have to consider the following 
situation. 

\begin{prop}\label{prop.transverse}
Let $Y$ be a submanifold with corners of $M$ and $X \subset Y$ be
a submanifold with corners of $Y$.
Then there is a unique embedding 
$[Y: X] \to [M:X]$ as a submanifold with corners such that
\begin{equation*}
\begin{array}{ccc}
 [Y:X]  &\rightarrow& [M:X] \\[3mm]
 \ \ \ \downarrow \beta_Y& &   \ \ \ \downarrow\beta_M\\[3mm]
 Y &\rightarrow& M 
\end{array}
\end{equation*}
commutes. The range of the embedding $[Y: X] \to [M:X]$ is the closure
of $Y \smallsetminus X$ in $[M : X]$.
\end{prop}

\begin{proof}
The statement of the proposition is essentially a local statement. Let us 
find good local models first. We assume $n=\dim X$, $n+\ell=\dim Y$
and $n+k=\dim M$. As described above $X$ is locally diffeomorphic to 
an open subset of $[0,\infty)^n$. The definition of submanifolds with corners
implies that $X$ does not meet boundary faces of $Y$ or $M$ of codimension 
$>n$. Thus any point $x\in X$ has an open neighborhood in $M$ where 
the iterated submanifold structure $X\subset Y\subset M$ is locally diffeomorphic to
\begin{equation*}
        [0,\infty)^n\times \{0\} \subset
    [0,\infty)^n\times \RR^\ell\times \{0\} \subset
    [0,\infty)^n\times \RR^k.
\end{equation*}
A more precise version of this is the following obvious lemma. Here $A\supcirc B$ 
stands for an open inclusion map (so $B$ is an open subset of $A$).

\begin{lemma}
Let $Y$ be a submanifold with corners of $M$ and $X \subset Y$ be
a submanifold with corners of $Y$.
Then any $x\in X$ has an open
neighborhood $U$ in $M$ such that there is a diffeomorphism $\phi:U\to
V$ to an open subset $V$ of $[0,\infty)^n\times \RR^k$ for which the
diagram
\begin{equation*}
\begin{matrix}
X&\supcirc & U \cap X & \cong & V \cap [0, \infty)^n \times \{0\}\hfill\\
\hookdownarrow&&\hookdownarrow&&\hookdownarrow\\
Y&\supcirc & U \cap Y & \cong & V \cap [0, \infty)^n \times {\mathbb R}^{\ell} \times \{0\}\hfill\\
\hookdownarrow&&\hookdownarrow&&\hookdownarrow\\
M&\supcirc & U &\cong & V \cap [0, \infty)^n \times \RR^{k}\hfill
\end{matrix}
\end{equation*}
commutes.
\end{lemma}

It is easy to see that Proposition~\ref{prop.transverse} holds for the
local model as the embedding ${S^{\ell-1}\times\{0\}\hookrightarrow
S^{k-1}}$ induces an embedding
\begin{eqnarray*} 
        [U\cap Y:U\cap X]&\cong & V \cap [0, \infty)^n \times
        S^{\ell-1} \times [0,\infty) \times \{0\}\\
        \hookrightarrow [U:U\cap X]&\cong& V \cap [0, \infty)^n \times
        S^{k-1} \times [0,\infty).
\end{eqnarray*}

The local embeddings thus obtained
then can be glued together using Lemma~\ref{lemma.loc} to get a global
map $[Y:X]\to [M:X]$. The other statements of the proposition are then obvious.
\end{proof}


\subsection{Iterated blow-up}\label{subsec.it.bu}

We now want to blow up a finite family of submanifolds.  

\begin{definition}
A finite set of connected submanifolds with corners 
$\maX = \{X_1,\ldots,X_k\}$, $X_i\neq \emptyset$,
of $M$ is said to be a {\em clean family of submanifolds} 
if, for any
indices $i_1,\ldots,i_t\in \{1,2,\ldots,k\}$, one has the following
properties:
\begin{itemize}
\item Any connected component of $\bigcap_{j=1}^t X_{i_j}$ is in $\maX$, 
that is, the family $\maX$ is closed under intersections.
\item For any  $x \in \bigcap_{j=1}^t X_{i_j}$ one has
$\bigcap_{j=1}^t T_x X_{i_j}= T_x\left(\bigcap_{j=1}^t
X_{i_j}\right)$.
\end{itemize}
\end{definition}

Examples:

\begin{enumerate}[(i)]
\item $M=\RR^6=\RR^3\times \RR^3$, $X_1:=\RR^3\times \{0\}$, $X_2:=
\{0\}\times \RR^3$, $X_3$ the diagonal of $\RR^3\times \RR^3$, $X_4:=\{0\}$. 
Then  $\maX := \{X_1,X_2,X_3,X_4\}$ is a clean
family.
\item Using the same notations as in (i), $\maX_0 := \{M,X_1,X_2,X_3,X_4\}$, 
$\maX_1 := \{M,X_1,X_2,X_4\}$ and $\maX_2 := \{M,X_1\}$ are also
clean families.
\item More generally, let $M$ be a vector space and $\maX = \{X_i\}$ a finite family
of affine subspaces closed under intersections. Then $\maX$ is a clean
family.
\end{enumerate}

If $\maX = \{X_i\}$ is a clean family of submanifolds and
the submanifolds $X_i$ are also {\em disjoint}, then we define $[M: \maX]$
by successively blowing up the manifolds $X_i$. The iteratively
blown-up space $[M: \maX] := [ \ldots [[M:X_1]: X_2] : \ldots : X_k]$
is independent of the order of the submanifolds $X_i$, as the blow-up
structure given by Lemma~\ref{lemma.loc} is local.

Let us consider now a general clean family $\maX$, and
let us define the new family $\maY := \{Y_\alpha\}$ consisting of the
{\em minimal} submanifolds of $\maX$ (\ie submanifolds that do not
contain any other proper submanifolds in $\maX$). By the assumption
that the family $\maX$ is closed under intersections, the family
$\maY$ consists of disjoint submanifolds of $M$. Let $M':=[M: \maY]$
be the manifold with corners obtained by blowing up the submanifolds
$Y_\alpha$.  Assuming that $\maY \neq \maX$, we set
$\maY_j:=\{Y\in \maY\,|\, Y\subset X_j\}$, for
$X_j\in \maX\setminus \maY$, and define $X_j':=[X_j:\maY_j]$.  By
Proposition~\ref{prop.transverse} $X_j'$ is the closure of
$X_j \smallsetminus \cup Y_\alpha$ in $M'$.  Let also $d_\maX$ be the
minimum of the dimensions of the minimal submanifolds of $\maX$ (\ie
the minimum of the dimensions of the submanifolds in $\maY$).  We then
have the following theorem.

\begin{thm}\label{thm.transverse} Assume $\maY \neq \maX$. Then, using the notation of the above paragraph, 
 the family~$\maX':= \{X'_j\}$ is a clean
family of submanifolds of~$M'$.  Moreover, the minimum dimension
$d_{\maX'}$ of the family $\maX'$ is greater that the minimum
dimension $d_{\maX}$ of the family $\maX$.
\end{thm}

\begin{proof} 
By Proposition \ref{prop.transverse}, the sets $X'_j$ are submanifolds
with corners of $M'$. Let $j_1 < j_2 < \ldots < j_t$ and let $Z' :=
X'_{j_1} \cap X'_{j_2} \cap \ldots \cap X'_{j_t}$. We first want to
show that $Z' \in \maX'$. Assume that $Z' \cap (M \smallsetminus
\bigcup Y_\alpha)$ is not empty. Then $Z := X_{j_1} \cap
X_{j_1} \cap \ldots \cap X_{j_1} \in \maX$ and hence $Z = X_i$, for
some $i$, by the assumption that $\maX$ is a clean
family. We only need to show that $Z' = X_i'$.

We have that $X_i \cap (M \smallsetminus \bigcup Y_\alpha) \subset X_{j_s}
\cap (M \smallsetminus \bigcup Y_\alpha)$, so $X_i' \subset X_{j_s}'$,
and hence $X_i' \subset Z' := \bigcap X_{j_s}'$.  We need now to prove
the opposite inclusion. Let $x \in Z'$. If $\beta(x) \not\in Y_\alpha$
for any $\alpha$, then $x=\beta(x) \in Z = X_i$ and hence $x \in
X_i'$. Let us assume then that $y := \beta(x) \in Y_\alpha$ for some
$\alpha$. By definition, this means that $x \in T_yM/T_yY_\alpha$ (and
is a vector of length one, but this makes no difference). Our
assumption is that $x \in T_y X_{j_s}/T_yY_\alpha$, for all $s$. But
our cleanness assumption  then implies $x \in T_y
X_i/T_yY_{\alpha}$, which means $x \in X_i'$, as desired.

It remains to prove that $TX_i' = \bigcap TX_{j_s}'$, where $X_i' = Z' =
X'_{j_1} \cap X'_{j_2} \cap \ldots \cap X'_{j_t}$, as above.  The
inclusion $TX_i' \subset \bigcap TX_{j_s}'$ is obvious. Let us prove the
opposite inclusion. Let then $\xi \in \bigcap T_xX_{j_s}'$, $x \in M' =
[M: \maY]$. If $\beta(x) \not\in Y_\alpha$, for any $\alpha$, then
$\xi \in TX_i'$, by the assumption that $\maX$ is a clean family. 
Let us assume then that $y := \beta(x) \in
Y_\alpha$.  Since our statement is local, we may assume that $Y_\alpha
= \RR^{n-s} \times [0,\infty)^s\times \{0\}$ and that $M
= \RR^{n-s}\times [0,\infty)^s\times \RR^k$. Then the tangent spaces
$T_y X_{j_s}$ identify with subspaces of $\RR^{n+k}$. Let us identify
$[M: Y_\alpha]$ with the set of vectors in $M$ at distance $\ge 1$ to
$Y_\alpha$. We then use this map to identify all tangent spaces to
subspaces of $\RR^{n+k}$. With this identification, $T_x X_j'$
identifies with $T_y X_j$. Therefore, if $\xi \in \bigcap T_x X_{j_s}'$,
then $\xi \in \bigcap T_y X_{j_s} = T_y X_i = T_x X_i'$.

For each manifold $X'_j$, we have $\dim X'_j = \dim X_j > \dim
Y_\alpha$, for some $\alpha$, so $d_{\maX'} > d_{\maX}$.
\end{proof}

We are ready now to introduce the blow-up of a clean family
of submanifolds of a manifold with corners $M$.

\begin{definition} \label{def.fam.blowUP}
Let $\maX = \{X_j\}$ be a non-empty clean family of
submanifolds with corners of the manifold with corners $M$. Let $\maY
= \{Y_\alpha\} \subset \maX$ be the non-empty subfamily of minimal
submanifolds of $\maX$. Let us define $M' := [M: \maY]$, which makes
sense since $\maY$ consists of disjoint manifolds. If $\maX = \maY$,
then we define $[M: \maX] = M'$. If $\maX \neq \maY$, let $d_{\maX}$
be the minimum dimension of the manifolds in $\maY$ and we define
$[M: \maX]$ by induction on $\dim(\maX) - d_{\maX}$ as follows. Let
$\maX' := \{X_j'\}$, where $X_j'$ is the closure of
$X_j \smallsetminus (\cup Y_{\alpha})$ in $M'$, provided that the
later is not empty (thus $\maX'$ is in bijection with
$\maX \smallsetminus \maY$). Then $\dim(M') - d_{\maX'} < \dim(M) - d_{\maX}$,
and $\maX'$ is a clean family of submanifolds with corners 
of $M'$, so $[M': \maX']$ is defined.  Finally, we define
\begin{equation*}
        [M: \maX] := [M': \maX'] = [[M: \maY] : \maX'].
\end{equation*}
\end{definition}

Another equivalent definition of $[M:\maX]$ is the following. Assume
$\maX=\{X_i\,|\, i=1,2,\ldots,k\}$. Then we say that $\maX$ is {\em
admissibly ordered} if, for any $\ell\in\{1,2\ldots,k\}$, the family
$\maX_\ell=\{X_i\,|\, i=1,2,\ldots,\ell\}$ is a clean
family as well, or equivalently, if it is closed under
intersections. After possibly replacing the index set and reordering
the $X_i$, any $\maX$ is admissibly ordered. Let us denote
$\maY:=\{X_1,\ldots, X_r\}$ for $r:=\# \maY$, with $\maY$ the
family of minimal submanifolds in $\maX$ as before, and $X_{r+1}$
corresponds to a submanifold $X'_{r+1}$ in the family $\maY'$ of
minimal submanifolds in $\maX'$. This gives the following iterative
description of the blow-up:
\begin{equation*}
        [M:\maX]=[[\ldots[M:X_1]:X_2]:\ldots:X_r]:X_{r+1}']:\ldots :X_k''']
\end{equation*}
where $'''$ stands for an appropriate number of $'$-signs.

For $\ell\in\{1,2\ldots,k\}$, let us then denote
\begin{equation*}
        M^{(\ell)}:=[[\ldots[M:X_1]:X_2]:\ldots:X_r]:X_{r+1}']:\ldots
        : X_\ell'''] \qquad Y^{(\ell)}:=X_\ell'''\subset M^{(\ell-1)}
\end{equation*} 
where
again $'''$ stands for an appropriate number of $'$-signs.  Then
$M=M^{(0)}$, $M^{(\ell)}=[M^{(\ell-1)}:Y^{(\ell)}]$ and
$M^{(k)}=[M:\maX]$.

\begin{definition}\label{def.canonical}
The sequences $Y^{(1)},Y^{(2)},\ldots, Y^{(k)}$ and $M^{(0)},
M^{(1)}, \ldots, M^{(k)}$ are called the {\em canonical sequences}
associated to $M$ and the admissibly ordered family $\maX$.
\end{definition}

Let $\beta_{\ell} : M^{(\ell)} = [M^{(\ell-1)}: Y^{(\ell)}] \to
M^{(\ell-1)}$ for $\ell\in\{1,2,\ldots,k\}$ be the corresponding
blow-down maps. Then we define the blow-down map $\beta :
[M: \maX] \to M$ as the composition
\begin{equation}\label{eq.it.blow-down}
        \beta
        := \beta_{1} \circ \beta_{2} \circ \ldots \circ \beta_{k} :
        M^{(k)}= [M: \maX] \to M = M^{(0)}.
\end{equation}

\section{Lie structure at infinity}\label{section.lie.struct}

Manifolds with a Lie structure at infinity were introduced
in \cite{aln1} (see also \cite{MelrosePCSL} for the general ideas
related to this definition). In this section, we consider the blow-up
of a Lie manifold by a submanifold with corners and show that the
blown-up space also has a Lie manifold structure.  To this effect, we
start with describing lifts of vectors fields to the blow-up.  By the
results of the previous section, we can then blow up with respect to a
clean family of submanifolds with corners. We also investigate the
effect of the blow-ups on the metric and Laplace operators (and
differential operators in general).

Let $M$ be a manifold with corners and let $\cB_M = \{H_1, \ldots,
H_k\}$ be its set of boundary hyperfaces. As usual, we define
\begin{equation}\label{eq.maV_M}
        \maV_M:=\left\{ V\in\Gamma(TM)\,|\, V|_H \mbox{ is tangent to
         $H$ },\; \forall H\in \cB_M\right\}.
\end{equation}
That is, $\maV_M$ denotes the Lie algebra of vector fields on $M$ that
are tangent to all boundary faces of~$M$. It is the Lie algebra of the
group of diffeomorphisms of $M$.

\subsection{Lifts of vector fields}
Let $M$ be a manifold with corners. As in the smooth case, we identify
the set $\Gamma(TM)$ of smooth vector fields on $M$ with the set of
derivations of $\CI(M)$, that is, the set of linear maps
$V: \CI(M) \to \CI(M)$ satisfying $V(fg) = fV(g) + V(f)g$. With this
identification, the Lie subalgebra $\maV_M \subset \Gamma(TM)$
identifies with the set of derivations $V$ that satisfy
$V(x_H\CI(M)) \subset x_H\CI(M)$, for all boundary defining functions
$x_H$ \cite{MelroseScattering}.

Let $M$ and $P$ be manifolds with corners and $\beta: P\to M$ a smooth,
surjective, map. Regarding vector fields as derivations, it is then
clear what one should mean by ``lifting vector fields from $M$ to
$P$,'' namely that the following diagram commutes
\begin{equation}
\begin{CD}
C^\infty(P)           @>{W} >>      C^\infty(P) \\     
@A{\beta^*}AA            @AA{\beta^*}A           \\
C^\infty(M)         @>{V}>>   C^\infty(M)         
\end{CD}
\end{equation}
where $\beta^*f=f\circ \beta$.  Given two vector fields $V$
on~$M$ and $W$ on~$P$ , we say that
$V$ \textit{lifts} to $W$ along $\beta$, if
$V(f)\circ \beta=W(f\circ\beta)$, for any $f\in C^\infty(M)$.
Considering the differential $\beta_*: T_pP\to T_{\beta(p)}M$, we then say 
that $V$ \textit{lifts to}  $W$ along $\beta$  if, and only if, $\beta_*W_p=
V_{\beta(p)}, \mbox{ for all } p\in P$.

For a vector field $W$ on $P$, $\beta_*W$ does not define
in general a vector field on $M$. If $W$ is the lift of a 
vector field $V$ on $M$, then $\beta_*W_p$ only depends on $\beta(p)$,
i.e.\ $\beta_*W_p= \beta_*W_q$ for all $p,q\in P$ with
$\beta(p)=\beta(q)$. We denote by
$\Gamma_\beta(TP)$ the set of all vector fields on $P$
that are lifts along $\beta$ of some vector field on $M$.  For any $W\in\Gamma_\beta(TP)$, the push-forward
$\beta_*W$ is well defined as a vector field on $M$.  By definition,
we have a map
\begin{equation}\label{beta* v flds}
      \beta_*: \Gamma_\beta(TP) \to \Gamma(TM),\quad
      (\beta_*W)_x:= \beta_*W_p, \: \beta(p)=x.
\end{equation}
 If $\beta$ is a diffeomorphism, then $\Gamma_\beta(TP)
= \Gamma(TP) $ and any vector field on $M$ can be lifted uniquely to
$P$. Note that $\Gamma_\beta(TP)$ is always a Lie subalgebra of
$ \Gamma(TP)$, since $\beta_*([W_1, W_2]_p) =
[\beta_*W_1, \beta_*W_2]_x $, if $\beta(p) = x$.

If $\beta$ is a submersion, then any vector field on $M$ lifts to $P$
along $\beta$, and the lift is unique mod $\ker \beta_*$, that is,
after fixing a Riemannian structure on $P$, there is an unique
horizontal lift $W$ such that $W_p\in (\ker \beta_*)^\perp$, $p\in P$.

\subsection{Lifts and products}

Let $P$, $M$ and $\beta$ as above. We assume in this subsection that
any vector field $V\in\Gamma(TM)$ has {\em at most one lift}
$W_V\in\Gamma(TP)$.  We now take product with a further manifold $N$
with corners. Then $T(M\times N) = TM \times TN$. Accordingly, a
vector field $\witi V\in \Gamma(T(M\times N))$ is then naturally the
sum of its $M$- and $N$-components: $\witi V(x,y)=\witi V_M(x,y)+\witi
V_N(x,y)$, $x\in M$, $y\in N$.

The following lemma answers when such a vector field lifts with
respect to $\beta\times \id:P\times N\to M\times N$.

\begin{lemma}\label{lem.prod}
Under the above assumptions (including uniqueness of the lift), any
vector field $\witi V\in \Gamma(T(M\times N))$ has a lift $\witi
W\in \Gamma(T(P\times N))$ if, and only if, for any $y\in N$, the
vector field $\witi V_M(\,.\,,y)\in \Gamma(TM)$ lifts to a vector
field $W_y$ on $P$.  In this case, the lift is $\witi
W(x,y)=W_y(x)+\witi V_N(x,y)$, in particular, the lift $\witi W$ is
uniquely determined.
\end{lemma} 

\begin{proof}
The only non-trivial statement in the lemma is to prove that the
vector field $\witi W$ defined by $\witi W(x,y)=W_y(x)+\witi V_N(x,y)$
is smooth, provided that the right hand side exists.

The uniqueness of the lift implies that the map
$\Gamma_\beta(TP)\to \Gamma(TM)$ is an isomorphism of vector spaces,
and thus its inverse, being a linear map, is a smooth map
$\Gamma(TM)\to \Gamma_\beta(TP)$, where we always assume the
$C^\infty$-Frechet topology in these spaces. The composition map
$Y\to \Gamma(TM)\to \Gamma_\beta(TP)$, $y\mapsto V_M(\,.\,,y)\mapsto
W_y$ is thus smooth as well. We have proven the smoothness of $\witi
W$.
\end{proof}

\subsection{Lifting vector fields to blow-ups}\label{subsec.lift.vf}

Let $M$ be a manifold with corners, $X$ a submanifold with corners.
We are interested in studying lifts of Lie algebras of vector fields
on $M$, tangent to all faces, along the blow-down map $\beta :
[M:X]\to M$.


For simplicity of presentation, {\em we shall restrict to the case
$\dim X<\dim M$}, in what follows (even if most of our results hold
for $\dim X=\dim M$).  We adopt from now on the convention
that \emph{any submanifold (with corners) is of smaller dimension than
its ambient manifold (with corners).}  The map $\beta$ is then
surjective and it yields a diffeomorphism
$[M:X] \smallsetminus \beta^{-1}(X) \to M \smallsetminus X$. The
problem of lifting vector fields thus is an extension problem, so the
lift is unique if it exists.  The uniqueness implies that lifts exist
on $M$ if and only if they exist on each open subset of $M$, i.e.\ the
lifting problem is a local problem.  Recall that $\maV_M$ was defined
in Equation~\eqref{eq.maV_M}.

In this subsection, we will show the following proposition on lifts of
vector fields to blow-ups.  A proof of this result can be found in
Section~5.3 of the unpublished manuscript \cite{melrose.dam}, so we
include a proof for completeness. Notice, however, that the extension
of this result to Lie manifolds is a new result, which is surprising,
in part, because it requires no additional assumptions on the Lie
manifold structure, see Subsection~\ref{subsec.bu-lm}.

\begin{prop}\label{prop.lift}
Let $M$ be a manifold with corners, $X$ a submanifold with corners,
and $V \in \maV_M$.  Then, there exists a
vector field $W \in \maV_{[M:X]}$ that lifts $V$ if, and only if, $V$ is tangent to~$X$.
\end{prop}

The proposition should be seen as an infinitesimal version of
Lemma~\ref{lemma.loc}.  Let us denote by $\Diffeo(M/X)$ the group of
diffeomorphisms of $M$ mapping $X$ onto itself. Then let
$\Diffeo(M):=\Diffeo(M/\emptyset)$.  In the case that $M$ is an open
subset of $[0,\infty)^n\times \RR^k$, and $X=M\cap
[0,\infty)^n\times \{0\}$, Lemma~\ref{lemma.loc} states that a Lie
group homomorphism $\alpha:\Diffeo(M/X)\to \Diffeo([M:X])$ exists such
that $\alpha(\phi)$ coincides with $\phi$ on $M\setminus X$. It thus
implies a Lie algebra homomorphism $\alpha_*$ between the
corresponding Lie algebras. The Lie algebra of $\Diffeo(M/X)$ consists
of those vector fields in $\maV_M$ whose restriction to $X$ is tangent
to~$X$. The Lie algebra of $\Diffeo([M:X])$ is $\maV_{[M:X]}$. The
image of $\alpha_*$ is $\Gamma_\beta(T[M:X])$.  As lifting vector
fields is a local property, these considerations already provide a
proof of Proposition~\ref{prop.lift}, assuming facts from the theory
of infinite-dimensional Lie groups and algebras.

In order to be self-contained we will also include a direct proof. As
before we will study a simple model situation first.

\begin{lem}\label{lem.blow-up.Rn}
Let $M=[0,\infty)^n \times \RR^k$ and
$X=[0,\infty)^n \times \{0\}\subset M$, and thus
$[M:X]=[0,\infty)^n\times S^{k-1}\times [0,\infty)$. Let $V\in \maV_M$
be a vector field that is tangent to $[0,\infty)^n\times \{0\}$, that
is we assume that $V$ is a vector field on $M$ tangent to the boundary
of $M$ and to the submanifold $X$.  Then there exists a lift of $V$ in
$\maV_{[M:X]}$, that is, there is a vector field $W\in \maV_{[M:X]}$
with $\beta_*W=V$ that is tangent to all boundary hyperfaces of
$[M:X]$.
\end{lem}

\begin{proof} 
At first, we assume $n=0$. Denoting $f_{\lambda}(x) =
f(\lambda x)$, a differential operator
$D\in \Diff(\RR^k\setminus \{0\})$ is homogeneous of degree $h$ if $(D
f)_{\lambda} = \lambda^h D f_\lambda$ for all $\lambda\in
(0,\infty)$. Radially constant vector fields on $\RR^k\setminus\{0\}$
thus define first order homogeneous differential operators homogeneous
of degree $-1$.

For $y=(y_1, ..., y_k)\in \RR^k\setminus 0$ (defining $X$) and $(r, \omega) \in
[0,\infty)\times S^{k-1}$, $x=\beta(r,\omega)=r\omega$, we can write
in polar coordinates, for $r\neq 0$,
\begin{equation}\label{cart to polar}
             \pa_{y_j}= \frac{\pa y_j}{\pa r}\pa_r +
             S_j(r)= \omega_j\pa_r + \frac{1}{r}S_j(1)
\end{equation}
where $S_j(r)$ is a vector field on $S^{k-1}$, depending smoothly on
$r\in(0,\infty)$. Note that since both $\pa_{y_j}$ and $\pa_r$ are
homogeneous of degree $-1$, the component $S_j$ is again of degree
$-1$, and this means $S_j(r)=\frac1r S_j(1)$ for all $r\in
(0,\infty)$. A vector field $V$ on $\RR^k$ vanishes at $0$ if, and
only if, it can be written as $V=\sum a_{ij}(y)y_i\pa_{y_j}$,
$x\in\RR^k$. Since $a_{ij}$ lifts to $\beta^*a_{ij}=
a_{ij}\circ \beta$ and since, writing $y=r\omega$,
\begin{equation}\label{lift_tg_X}
             y_i\pa_{y_j}= r\omega_i\omega_j\pa_r + \omega_iS_j(1)
\end{equation}
clearly extends to $r=0$, we have that $V$ lifts to $[\RR^k: 0]$ and
it is tangent to $S^{k-1}$ at $r=0$. The statement for $n=0$ follows.
The case for general $n$ then follows from Lemma~\ref{lem.prod}.
\end{proof}

Now, as the existence of a lift is a local property,
Lemma~\ref{lem.blow-up.Rn} also holds if $M$ is an open subset of
$[0,\infty)^n\times \RR^k$ with $X=M\cap [0,\infty)^n\times\{0\}$. If
$M$ is a manifold with corners and if $X$ is submanifold with corners
of it, then we obtain that a vector field on $M$ can be lifted in any
coordinate neighborhood, if it is tangent to $X$.  As the lifts are
unique we obtain sufficiency in Proposition~\ref{prop.lift} by gluing together the
local lifts.  Note that we obtain from Equation~\eqref{lift_tg_X} that
lifts of vector fields tangent to $X$ are in fact tangent to the
fibers of $\beta^{-1}(X)=S^MX\to X$.

It  also follows from (\ref{cart to polar}) that a vector field
$V\in \Gamma(TM)$ for which $V|_X$ is not tangential to $X$ does not
lift to a vector field in $\maV_{[M:X]}$, so we finish the proof of Proposition \ref{prop.lift}.

We now choose a true Riemannian metric $\bar g$ on $M$ (\ie
smooth up to the boundary). In contrast to the $\maV$-metric,
introduced later, this is a metric in the usual sense, i.e.\ a smooth
section of $T^*M\otimes T^*M$ which is pointwise symmetric and
positive definite.  Recall that we denoted by $r_X : M \to
[0,\infty)$ a smoothed distance function to $X$, that is, a
continuous function on $M$, smooth outside $X$ that close to $X$ is
equal to the distance function to $X$ with respect to $\bar g$ and
$r_X^{-1}(0) = X$.

\begin{cor}\label{cor.lift}
Let $M$ be a manifold with corners, $X$ a submanifold with corners,
and $r_X : M \to [0,\infty)$ be a smoothed distance function to
$X$. Let $V \in \maV_M$.  Then there exists a vector field
$W \in \maV_{[M:X]}$ such that $W = r_XV$ on $M \smallsetminus
X \subset [M:X]$.
\end{cor}

\begin{proof}
Again, it is sufficient to check the lifting property locally.  We
assume that $U$ is open in $M$ and that $y_1,\ldots,y_k$ are functions
defining $X$ as in Definition~\ref{def.subm} (i). We can assume that
$r_X^2=\sum_i y_i^2$. We then can write
\begin{equation}\label{lift_rXV}
        r_XV=\sum_i \frac{y_i}{r_X} \;y_i V.  
\end{equation}
Proposition~\ref{prop.lift} says that the vector fields $y_iV$ lift to
$\Gamma(T[M:X])$ as vector fields tangent to the faces. The functions
$\frac{y_i}{r_X}$, defined a priori on $U\setminus (U\cap X)$, extend
to smooth functions on $\beta^{-1}(U)$. Thus $r_XV$ has a lift locally
on $U$, and by uniqueness of the local lifts, these lifts match
together to a global lift.
\end{proof}

If $X$ is connected, then $\{r_X\}\cup \{x_H\,|\,H\in \maB\}$ is a set
of boundary defining functions for $[M:X]$, where each $x_H$ is the
defining function for the hyperface $H$ of $M$. Furthermore
$W\in\maV_{[M:X]}$ if, and only if, $W(x_Hf)=x_H\tilde f$ and
$W(r_Xf)=r_X\tilde f$ (where we are actually considering lifts of
$x_H$ and $r_X$ to $[M:X]$). For non-connected $X$, the distance to
$X$ has to be replaced by the distance functions to the connected
components in the obvious way, and the same result remains true.

The set of vector fields in $\maV$ which are tangent to $X$ forms a
sub-Lie algebra of $\maV$ which is also a $\CI(M)$-submodule.  This is
the Lie-algebra of $\Diffeo(M/X)$.  Inside this sub-Lie algebra, the
vector fields \emph{vanishing} on $X$ form again a sub-Lie algebra,
which is again a $\CI(M)$-submodule. This is the Lie algebra to the
group $\Diffeo(M;X)$ the Lie group of diffeomorphisms of $M$ that fix
$X$ pointwise.



We can characterize the lifts of such vector fields. Let $V\in \maV_M$
with lift $W\in \maV_{[M:X]}$. It follows from the definition that
$ \beta_*(W(p))= V_{\beta(p)}$. Hence, $V|_X\equiv 0$ is equivalent to
\begin{equation}
  \beta_*(W(p))=0\qquad \forall p\in \beta^{-1}(X).
\end{equation}
We obtain that $V$ vanishes on $X$ if, and only if, $W|_{S^MX}$ is a vector
field on $\beta^{-1}X=S^MX\subset \pa [M:X]$ which is tangent to the
fibers of $S^MX\to X$.  With (\ref{lift_rXV}) we see that lifts of
vector fields $r_XV$ from $M\setminus X$ to $[M:X]$ are also tangent
to these fibers.

\subsection{Lie manifolds}\label{subsec.liemanifolds}

Let us recall the definition of a Lie manifold and of its Lie
algebroid \cite{aln1,ammann.lauter.nistor:07}.  Let $M$ be a compact
manifold with corners. We say that a Lie subalgebra $\maV \subset
{\maV_M}$ is a {\em structural Lie algebra of vector fields} if it is
a finitely generated, projective $C^\infty(M)$-module.  The
Serre-Swan theorem then yields that there exists a vector bundle
$A$ satisfying $\maV\cong \Gamma(A)$.  In particular, $\Gamma(A)$ is a Lie algebra. Moreover, 
\begin{enumerate}
  \item there is a map $\rho: {A} \to TM$, called the \emph{anchor map}, which induces the inclusion map
$\rho: \Gamma({A}) \to \Gamma(TM)$;
  \item $\rho$ is a Lie algebra homomorphism and $[V,fW]=f[V,W]+\left(\rho(V)f\right)W$.
\end{enumerate}
The vector bundle $A$ is then what is called  a \emph{Lie algebroid}.

\begin{defn} 
A \textit{Lie manifold} $M_0$ is given by a pair $(M,\maV)$ where $M$
is a compact manifold with corners with $M_0=int(M)$, and $\maV$ is
structural Lie algebra of vector fields such that
$\rho_{|M_0}:{A}|_{M_0}\to TM_0$ is an
isomorphism. A \textit{$\maV$-metric} is a smooth section of
${A}^*\otimes {A}^*$ which is pointwise symmetric and positive
definite.
\end{defn}

A $\maV$-metric defines a Riemannian metric on the interior $M_0$ of
$M$.  If $\maV$ is fixed, then any two such metrics are bi-Lipschitz
equivalent.  The geometric properties of Riemannian Lie manifolds were
studied in~\cite{aln1}. It is known that any such $M_0$ is necessarily
complete and has positive
injectivity radius by the results of Crainic and
Fernandes \cite{CrainicFernandes}.

To avoid a misunderstanding, we emphasize that the metric $\bar g$
introduced in Subsection~\ref{subsec.lift.vf}, and used to define smoothed 
distance functions, is not a
$\maV$-metric. The metric $\bar g$ extends to the boundary as a smooth
section of $T^*M\otimes T^*M$, whereas a $\maV$-metric does not. One
can also use the terminology that $\bar g$ is a true metric on $TM$,
whereas $\maV$-metrics are usually called metrics on $A$.

To each Lie manifold we can associate an algebra of
$\maV$-differential operators $\Diff_\maV(M)$, the enveloping algebra
of $\maV$, generated by $\maV$ and $C^\infty(M)$. If $E, F$ are
vector bundles over $M$, then we define $\Diff_\maV(M;E,F):=
e_FM_N(\Diff_\maV(M))e_E,$ where $e_E, e_F$ are projections onto $E,
F\subset M\times \CC^N$.

It is shown in \cite{aln1} that all geometric differential
operators associated to a compatible metric on a Lie manifold
are $\maV$-differential, including the classical Dirac operator and
other generalized Dirac operators.  In particular, the de Rham
differential defines an operator
$d: \Gamma(\bigwedge^qA^*)\to \Gamma(\bigwedge^{q+1}A^*)$ and
$d\in \Diff^1_\maV(M;\bigwedge^{q}A^*,\bigwedge^{q+1}A^*)$, and its
formal adjoint $d^*$ is an operator in
$\Diff^1_\maV(M;\bigwedge^{q+1}A^*,\bigwedge^{q}A^*)$. By composition, 
we know that the \textit{Hodge-Laplace operator}
\begin{equation}
        \Delta : =(d+d^*)^2=dd^*+d^*d \in \Diff^2_\maV(M;
        {\textstyle\bigwedge\nolimits^q}A^*),
\end{equation}
is thus $\maV$-differential as well. It is moreover elliptic in that algebra, in
the sense that its principal symbol, a function defined on $A^*$, is
invertible, see \cite{aln1}.

We shall need the following regularity result from \cite[Theorem~5.1]{ain}.
The Sobolev space $H^{k}(M,\maV)$ associated to a Lie manifold
$(M,\maV)$ with a $\maV$-metric $g$ on its Lie algebroid $A$ is defined
in \cite{ain} as
\begin{equation}
        H^{k}(M, \maV) := \{u: M \to \CC\,|\; V_1 \ldots V_j u \in
        L^2(M, d\vol_g)\; \forall V_1, \ldots, V_j \in \maV, \, j \le
        k\,\}
\end{equation}
Note that these Sobolev spaces are not the Sobolev spaces with respect
to the euclidean metric, but with respect to the blown-up metric
$g$, and they depend only on the Lie manifold structure defined
by $\maV$.

\begin{thm} 
\label{thm.ain}Let $m\in \ZZ^+$, $s\in\ZZ$.
Let $P \in \Diff^m_\maV(M, \maV)$ be elliptic and $u \in H^r(M, \maV)$ be 
such that $Pu \in H^s(M, \maV)$. Then $u \in H^{s+m}(M, \maV)$. The same result holds
for systems.
\end{thm}
\onlyarxive{An important example of a Lie manifold is when $A$ is Melrose's 
$b$-tangent space. This leads to the $ b$-calculus. This example is carried out 
in Appendix~\ref{app.ex}.}

\subsection{Blow-up of Lie manifolds}\label{subsec.bu-lm}

Let $M$ carry a Lie manifold structure, and $X$ be a submanifold with corners of
$M$. We want to define a Lie structure on $[M:X]$. 

We begin by choosing a true metric $\bar g$ on $TM$, that is, $\bar g$ is smooth up to the
boundary.  Let $U_\epsilon(X)$ be an $\epsilon$-neighborhood of $X$ in $M$
with respect to $\bar g$. Later on we will need that the distance
function to $X$ with respect to $\bar g$ is a smooth function on
$U_\epsilon(X)\smallsetminus X$ for sufficiently small
$\epsilon>0$. Unfortunately, such an $\epsilon>0$ does not exists for
arbitrary metrics $\bar g$ on $M$. On the other hand, such an
$\epsilon>0$ exists if a certain compatibility condition between $M$,
$X$ and $\bar g$ holds, and for given $M$ and $X$ a compatible $\bar
g$ exists. More precisely, the compatibility condition is that
there is an $\epsilon>0$ such that for any $V\in T_xM$, $x\in X$,
$V\perp T_xX$, the curve $\gamma_V:t\mapsto \exp_x(tV)$ is defined for
$|t|< \epsilon$ and the boundary depth is constant along such
curves. For example metrics $\bar g$ whose restriction to a tubular
neighborhood of $X$ are product metrics of $\bar g|_X$ with a metric
on a transversal section, satisfy this compatibility
condition. However, we cannot assume without loss of generality that
for given $M$ and $X$ there is a metric $\bar g$ providing such a
product structure. (For example, consider the case that the normal
bundle of $X$ in $M$ is non-trivial. Then there is no product metric
on a neighborhood of $X$, whereas a compatible metric exists.)

Now let $r_X$ denote the smoothed distance function to $X$ with
respect to a true metric $\bar g$ that satisfies the compatibility
condition of the previous paragraph. The function $r_X$ thus
coincides with the distance function to $X$ on $U_\epsilon(X)$, for
some $\epsilon>0$, and is smooth and positive on $M\setminus X$.
We will also assume $r_X\leq 1$.

Any $x\in X$ has an open neighborhood $U$ in $M$ and a submersion
$y=(y_1,\ldots,y_k):U\to \RR^k$ with $X\cap U=y^{-1}(0)$ and
$r_X=|y|=\sqrt{\sum_i y_i^2}$.

\begin{lemma}\label{lemma.LieMan}
Let $(M,\maV)$ be a Lie manifold, $X \subset M$ be a submanifold with
corners. Then
\begin{equation*}
        \maV_0 := \bigl\{\sum f_i V_i\,|\, f_i\in C^\infty(M), \quad
           f_i|_X\equiv 0, \quad V_i\in \maV\bigr\}
\end{equation*} 
is a $C^\infty(M)$-submodule and a Lie subalgebra of $\maV$.  The lift
\begin{equation*}
        \maW_0
        := \{W \in \Gamma_\beta(T[M:X])\,|\, \beta_*(W) \in \maV_0 \}
\end{equation*}
is isomorphic to $\maV_0$ as a $\CI(M)$-module and as a Lie algebra. 
Let $\maW$ be the $\CI([M:X])$-submodule of $\maV_{[M:X]}$
generated by $\maW_0$, i.\ e.
\begin{equation*} 
        \maW := \bigl\{\sum_i f_i W_i\,|\, f_i \in \CI([M:X]),\,
        W_i \in \maW_0\, \bigr\}.
\end{equation*} 
Then, for any vector field $W\in \maW$, its restriction $W|_{S^MX}$ is
tangent to the fibers of $S^MX$ and $\maW$ is closed under the
Lie bracket. 
\end{lemma} 

\begin{proof}
The vector space $\maV_0$ is a Lie subalgebra of $\maV_{[M:X]}$ as
\begin{equation*}
        [f_1V_1, f_2V_2] = f_1f_2[V_1,V_2] + f_1V_1(f_2)V_2 -
        f_2V_2(f_1)V_1.
\end{equation*}
Incidentally, the same equation shows that $\maW$ is closed under the
Lie bracket.

By Propositon~\ref{prop.lift}, any vector field in $\maV_0$ can be
lifted uniquely and smoothly to the blow-up. The map
$\beta_*: \Gamma_\beta(T[M:X]) \to \Gamma(TM)$ is obviously an
isomorphism of $\CI(M)$-modules and of Lie algebras.  Then $\maW_0$ is
a Lie algebra of vector fields in $\maV_{[M:X]}$, and so is $\maW$.
It follows from the definition of lift that $W|_{S^MX}$ is tangent to the
fibers for all $W\in \maW$ (see the remarks at the end of Section \ref{subsec.lift.vf}).
\end{proof}

\begin{lemma}\label{equiv_def_W}
Let $(M,\maV)$ be a Lie manifold, $X \subset M$ be a submanifold with
corners. Let $r_X$ be a smoothed distance function to $X$. Then
\begin{equation*}
        \maW_1 := \{W \in \Gamma(T[M:X])\,|\, \exists V\in \maV
        \mbox{ with } W|_{M\setminus X} =  r_X V|_{M\setminus X} \}
\end{equation*}
is isomorphic to $\maV$ as a $\CI(M)$-module.  
Furthermore
the natural multiplication map 
\begin{equation*}
        \mu: \CI([M:X])\otimes_{\CI(M)} \maW_1 \to \maW \subset \maV_{[M:X]}
\end{equation*}
is an isomorphism of $\CI([M:X])$-modules, and hence $\maW$ is a
projective $\CI([M:X])$-module.
\end{lemma}

\begin{remark} The previous two lemmata imply that there are
surjective linear maps $\CI([M:X]) \otimes_{\CI(M)} \maW_i\to \maW$
for $i=0,1$. As stated above, the resulting map for~$i=1$ is an
isomorphism. However, one can show that the resulting map is not
injective for $i=0$.
\end{remark}

Often $W\in \maW\subset \maV_{[M:X]}$ will be identified in notation
with $W|_{M\setminus X}$ and with $\beta_*W\in \maV_{M}$ if it exists.
(Recall that $\maV_M$ was defined in Equation~\eqref{eq.maV_M}.)

\begin{proof}[Proof of Lemma~\ref{equiv_def_W}]  
Let us denote $P := [M:X]$, to simplify notation.  The map
$\maV\to \maW_1$, which associates to a vector field $V\in \maV$ a lift
of $r_X V$, is obviously an isomorphism of $\CI(M)$-modules.

Now, we will show $\maW_1\subset \maW$. This means that for
$V\in \maV$ we will show that $r_XV$ lifts to a vector field in
$\maW$. With a partition of unity argument we see that without loss of
generality we can assume that the support of $V$ is contained in an
open set $U$, such that a function $y:U\to \RR^k$ as above exists. We
choose $\chi \in \CI(M)$ with support in $U$ and such that $\chi\equiv
1$ on the support of $V$.  We then write
\begin{equation*}  
                   r_XV= \sum_i \frac{\chi y_i}{r_X} \;\chi y_i V.
\end{equation*}
Since $\chi y_iV\in \maV_0$ and $\chi y_i/r_X\in \CI(P)$, the
assertion follows.

In order to show that $\maW_1$ generates $\maW$, we take a function
$f\in \CI(M)$, vanishing on $X$, and $V\in \maV$.  We have to show
that $fV$ is in the $\CI(P)$-module spanned by $\maW_1$.  Similarly to
above, we can assume that the support of $f$ is in an open set $U$,
such that $y$ exists on $U$.  We then can write $f=\sum h_iy_i$ with
$h_i\in C^\infty(M)$ and support in $U$.  We write
\begin{equation*}
        fV= \sum \frac{h_iy_i}{r_X}\, r_X V.
\end{equation*}
The vector field $r_XV$ lifts to a vector field in
$\maW_1$. Since $\frac{y_i}{r_X}\in \CI(P)$, the claim that
$\maW_1$ generates $\maW$ follows.

Finally, to prove that the multiplication map
$\mu: \CI(P)\otimes_{\CI(M)} \maW_1 \to \maW$ is an isomorphism of
$\CI(P)$-modules, it is enough to show $\mu$ is injective (since we
have just proved that it is surjective). Using the isomorphism from
above $\maW_1 = r_X \maV \simeq \maV$ as $\CI(M)$-modules. Hence by
the projectivity of $\maV$ as a $\CI(M)$--module, we can choose an
embedding $\iota : \maW_1 \to \CI(M)^N$ with retraction
$\CI(M)^N \to \maW_1$, where both $\iota$ and $r$ are morphisms of
$\CI(M)$-modules and $r \circ \iota = id$, the identity. The embedding
$\iota$ corresponds to an embedding $j : A \to \RR^N$ of vector
bundles. By definition, $A\vert_{M \smallsetminus X} = TM
\vert_{M \smallsetminus X}$.  We can therefore identify the restrictions
of the vector fields in $\maW$ to sections of
$A\vert_{M \smallsetminus X}$, which then yields an embedding
$\iota_0: \maW \hookrightarrow \Gamma(M \smallsetminus
X, \RR^N)=\CI(M\smallsetminus X)^N$.  Let us denote by $res$ the
restriction from $P$ to $M\setminus X$. We thus obtain the diagram
\begin{equation}
\begin{CD}
\CI(P) \otimes_{\CI(M)} \maW_1 @>{\mu} >> \maW \\
@V{\id \otimes \iota}VV @VV{\iota_0}V \\ 
\CI(P) \otimes_{\CI(M)}\CI(M)^{N} @>{res}>> \CI(M \smallsetminus X)^{N} \\
@V{=}VV @VV{=}V \\ 
\CI(P) ^{N} @>{res}>> \CI(M \smallsetminus X)^{N}
\end{CD}
\end{equation}
This diagram is commutative by the definition of $i_0$.

We have that $(id \otimes r) \circ (id \otimes \iota) = id$, and hence
$id \otimes \iota$ is injective. Moreover, all the other vertical maps
and the restriction maps are injective. It follows from the
commutativity of the diagram that $\mu$ is injective as well.
\end{proof}

In the following we write $r_X\maV$ for $\maW_1$, and for $\maW$ which
is the $\CI(P)$-module generated by it, with $P:=[M:X]$,  we also write $\CI(P)r_X\maV$.
We obtain

\begin{thm}\label{theorem.Lie.P}
Let $(M,\maV)$ be a Lie manifold, $X \subset M$ be a submanifold with
corners, and $r_X$ be a smoothed distance function to~$X$. Denote by
$P:=[M:X]$ the blow-up of $M$ along $X$. Then the $\CI(P)$-module
$\maW := \CI(P) r_X \maV$ defines a Lie manifold structure
on~$P$, which is independent of the choice of $r_X$.
\end{thm}

\begin{proof} Clearly $\maW$ consists of vector fields. 
The previous lemma shows that $\maW$ is a projective $C^{\infty}(P)$--module.
Proposition \ref{prop.lift} shows that $\maW \subset \maV_{P}$, that
is, that $\maW$ consists of vector fields tangent to all faces of $P$
(Equation~\eqref{eq.maV_M}). Lemma~\ref{lemma.LieMan} shows that $\maW$
is a Lie algebra (for the Lie bracket). Moreover, if $V$ is any vector
field on the interior $P$ and $U$ is an open set whose closure does
not intersect the boundary of $P$, then there exits $V_0 \in \maV$
such that $V_0 = r_X^{-1}V$ on $U$. Then $r_X V_0 \in \maW$ restricts
to $V$ on $U$. This shows that there are no restrictions on the vector
fields in $\maW$ in the interior of $P$. This completes the proof.
\end{proof}

\subsection{Direct construction of the blown-up Lie-algebroid}

We keep the notation of the previous subsection, especially of
Theorem \ref{theorem.Lie.P}. In particular, let $X \subset M$ be a
submanifold with corners. Since $\maW$ (introduced in
Theorem \ref{theorem.Lie.P}) is projective, there is a Lie algebroid
$B$ over $[M:X]$ such that $\maW$ is isomorphic to $\Gamma(B)$ as
$\CI([M:X])$-modules and Lie algebras. We now provide a direct
construction of $B$. We will denote by $T^{bX}[M:X]$ the
vector bundle whose sections are the vector fields on $[M:X]$ tangent
to all the faces obtained by blowing up $X$ in $M$.

In the following we will always use a smoothing $r_X$ of the distance
function to $X$, and we again assume $r_X$ takes values in $[0,1]$.
Different choices of metrics $\bar g$ or different smoothing will
provide different functions $r_X$. However, if $r_X'$ comes from other
choices than~$r_X$, then there is a constant $C>0$ with $C^{-1}\leq
r_X'/r_X \leq C r_X$ and due to compactness all derivatives of
$r_X'/r_X$ are bounded. We start with a preparatory lemma.

\begin{lemma}\label{lemma.b-iso}
Let $X$ be a submanifold of $M$, and $r_X$ be a smoothed distance
function to~$X$.  Then the map $T(M\setminus X)\to T(M\setminus X)$,
$V\mapsto r_X^{-1}V$ extends to a bundle isomorphism
\begin{equation*}
        \kappa : T^{bX}[M:X]\to \beta^* TM.
\end{equation*}
\end{lemma}
The proof is straightforward. Note that $\kappa$ is not the map
$\beta_*:T^{bX}[M:X]\to \beta^* TM$, but we have $\beta_*=r_X\kappa$.

As a vector bundle we then simply define
\begin{equation*}
        B:=\beta^* A=\{(V,x)\in A\times [M:X]\,|\,V\in A_{\beta(x)}\}.
\end{equation*} 
The anchor map $\rho_A:A\to TM$ pulls back to a map
$\beta^*\rho:\beta^* A\to \beta^*TM$, and we define the anchor
$\rho_B$ of $B$ to be the composition
\begin{equation*}
        B=\beta^*A\stackrel{\beta^*\rho_A}{\longrightarrow} 
        \beta^*TM \stackrel{\kappa^{-1}}{\longrightarrow}
        T^{bX}[M:X]\longrightarrow T[M:X]
\end{equation*}
In order to turn $B$ into a Lie algebroid, one has to specify a
compatible Lie bracket on sections of $B$. The Lie bracket $[.,.]_A$
on $\Gamma(A)$ will not be compatible with the previous structures.
However the Lie bracket $[.,.]_B$ given by
\begin{equation*}
        [V,W]_B:= r_X[V,W]_A + (\pa_Vr_X)W - (\pa_W r_X)V,
\end{equation*}
for all $V,W\in\Gamma(A)\stackrel{\beta^*}\hookrightarrow \Gamma(B)$
can be extended in the obvious way to $\Gamma(B)$, and this bracket is
compatible in the following sense:
\begin{enumerate}[{\rm (a)}]
\item $[f_1W_1,f_2W_2]_B=f_1f_2[W_1,W_2]_B+f_1(\pa_{\rho_B(W_1)}f_2)W_2
-f_2(\pa_{\rho_B(W_2)}f_1)W_1$
\item The map $\Gamma(B)\to \Gamma(T[M:X])$ induced by $\rho_B$
is a Lie-algebra homomorphism.
\end{enumerate}

One checks that $\Gamma(B)=\maW$.

\begin{remark}
The constructions in this section depend on $r_X$, and thus on the
choices of~$\bar g$ and the smoothing. Let $r_X'$ be a different
choice of a function with the properties of $r_X$. Using $r'_X$
instead of $r_X$ will lead to a different $\kappa'$ $B'$, and $\rho'$
replacing $\kappa$, $B$, and $\rho$. However, the new choices only
differ by a $r_X'/r_X$-factor from the old ones. In particular the
bundles $B$ and $B'$ thus obtained are isomorphic.
\end{remark}

\subsection{Geometric differential operators on blown-up manifolds}

We now study the relation between the Laplace operator
on $M$ and the one on $[M:X]$.

\begin{prop}\label{prop.these.are.lifts}
Let $(M,\maV)$ be a manifold with a Lie structure at infinity,
$\maV=\Gamma(A)$, for some vector bundle $A \to M$.  Assume that $M$
carries both a $\maV$-metric $g$ on $A$, and a true metric $\bar g$ on
$TM$ which is compatible with a submanifold $X$ of $M$ in the sense of
subsection~\ref{subsec.bu-lm}. Let $r_X$ denote a smoothed distance
function to $X$ with respect to the metric $\bar g$. Then
\begin{equation*}
        \grad_g r_X^2\in\maW
\end{equation*}
or more exactly the vector field $\grad_g r_X^2\in \Gamma(A)$ has a
unique lift in $\maW$.  Furthermore $\|\grad_g r_X\|^2\in \CI([M:X])$.
\end{prop}

\begin{proof}
We write $r_X^2\in \CI(M)$ locally as $\sum_i y_i^2$.  As $g$ is a
metric on $A$, it is fiberwise non-degenerate so it also defines a
metric $g^b$ on $A^*$. This dual metric $g^b$ is locally given by
$\sum_i e_i\otimes e_i$ where $e_i$ is a local $g$-orthonormal frame,
and is a section of $A\otimes A$. Let $\rho : A \to TM$ be the anchor
map of $A$. The dual map of $\rho$, \ie fiberwise composition with 
$\rho$,  yields a smooth map $\rho^*:T^*M\to A^*$, $T_p^*M \ni \alpha\mapsto \alpha\circ \rho\in A^*$. 
The contraction $T^*M \to A$ of
this map with $g^b$ will be denoted as $T^*M\ni\alpha\to \alpha^\#\in
A_p^*$. The $g$-gradient of a smooth function is by definition $\grad_g
f:= (df)^\#\in \Gamma(A)$. Thus we have
\begin{equation*}
        \grad r_X^2=(d r_X^2)^\# = 2 \sum_i y_i (dy_i)^\#.
\end{equation*}
Obviously the last equation only holds locally. From the remarks above
one sees that $(dy_i)^\#=\grad_g y_i$ is a local section of $A$, and
thus using Lemma~\ref{lemma.LieMan} it we see that $y_i \grad_g y_i$
lifts to $\maW$.  This implies that $\grad_g r_X^2$ locally lifts to
$\maW$, and thus globally.

The proof of the second statement is a bit subtle. The first subtle
point is that $\|\grad_g r_X\|^2$ is not well-defined as a function on
$M$, but only as a function on $[M:X]$.  The second subtle point is
that the Gauss lemma does not provide $\|\grad_g r_X\|^2=1$ close to
$X$ as $r_X$ is a smoothed distance with respect to the metric $\bar
g$, whereas the gradient is taken with respect to $g$.

However the Gauss lemma (applied for the metric $\bar g$) does provide
that $dr_X$ is a well-defined smooth function $[M:X]\to T^*M$
commuting with the maps to $M$. Thus $\rho^*\circ
dr_X \otimes \rho^*\circ dr_X$ is a smooth function $[M:X]\to
A^*\otimes A^*$. The contraction with $g^b\circ \beta$ then yields
$\|\grad_g r_X\|^2=\|d r_X\|^2 \in \CI([M:X])$.
\end{proof}

Let us now examine the effect of blow-up on Sobolev spaces.  Recall
that the Sobolev space $W^{k,p}(M,\maV)$ associated to a Lie manifold
$(M,\maV)$ with a $\maV$-metric $g$ on its Lie algebroid $A$ is defined
in \cite{ain}
\begin{equation}
        W^{k,p}(M, \maV) := \{u: M \to \CC\,|\; V_1 \ldots V_j u \in
        L^p(M, d\vol_g)\; \forall V_1, \ldots, V_j \in \maV, \, j \le
        k\,\}
\end{equation}

\begin{lemma}\label{lemma.sobolev}
Using the notation of the Lemmma \ref{lemma.LieMan}, we have
\begin{equation*}
        W^{k,p}([M:X], \maW) = \{u: M\to \CC\,|\; r_X^{j} V_1 \ldots
        V_j u \in L^p(M, d\vol_g)\; \forall V_1, \ldots,
        V_j \in \maV, \, j \le k\,\}
\end{equation*}
\end{lemma}

\begin{proof}
We have that $M$ and $[M:X]$ coincide outside a set of measure zero,
hence we can replace integrable functions on $[M:X]$ by functions on
$M$ integrable over $M\setminus X$.  The result for $k=1$ follows from
Lemma \ref{equiv_def_W}; for $k>1$, use induction on $k$ together with
the fact that $V_i r_X - r_X V_i = V_i(r_X) \in \CI([M:X])$ is a
bounded function, so that $(r_XV_i)(r_XV_j)u=r_X^2V_iV_ju
+V_i(r_X)r_XV_ju \in L^p(M\setminus X)$.
\end{proof}

Let us record also the effect of the blow-up on
metrics and differential operators.

\begin{lemma}\label{lemma.metric1}
We continue to use the notation of Lemmas~\ref{lemma.LieMan}
and~\ref{equiv_def_W}, in particular, $r_X$ is a smoothed distance
function to $X$. Let $A \to M$ be the Lie algebroid associated to
$\maV$, so that $\maV \simeq \Gamma(A)$.  Let us choose a metric $g$
on $A$. Let~$B$ be the Lie algebroid associated to
$([M:X], \maW)$. Then the restriction of $r_X^{-2} g$ to
$M \smallsetminus X$ extends to a smooth metric~$h$ on~$B$. Let
$\Delta_g$ and $\Delta_h$ be the associated Laplace operators.  Then
the operator
\begin{equation*}
       u\mapsto  D(u):=r_X^{\frac{n+2}2} \Delta_g (r_X^{-\frac{n-2}2} u) 
       - \Delta_h u. 
\end{equation*}
is given by multiplication with a smooth function on $[M:X]$, that is
$D \in \Diff^0_\maW([M:X])$.
Furthermore
\begin{equation*}
        r_X^2 \Delta_g - \Delta_h \in \Diff^1_{\maW}([M:X]). 
\end{equation*}
In particular, $r_X^2 \Delta_g$ is elliptic in
$\Diff^2_{\maW}([M:X])$.
\end{lemma}

\begin{proof}The operator $r_X^{\frac{n+2}2} \Delta_g r_X^{-\frac{n-2}2}$ and   
$\Delta_h$ have the same principal symbol, are  symmetric with respect 
to $d \vol_h$, and are smoth differential operators on $[M:X]$.
Thus $D$ is in $C^\infty([M:X])= \Diff^0_\maW([M:X])$.

%

Applying the formula $\Delta(uv)=v\Delta u+ u\Delta v + 2g(\grad_g
u,\grad_g v)$ we obtain
\begin{eqnarray*}
        {r_X}^{\frac{n+2}2} \Delta_g ({r_X}^{-\frac{n-2}2}u)&=&
        r_X^2 \Delta_g u+ {r_X}^{\frac{n+2}2} (\Delta_g
        {r_X}^{-\frac{n-2}2})u -2 \frac{n-2}2 {r_x}(\grad r_X)(u)\\
        &=& r_X^2 \Delta_g u+ {r_X}^{\frac{n+2}2} (\Delta_g
        {r_X}^{-\frac{n-2}2})u - \frac{n-2}2 (\grad r_X^2)(u)
\end{eqnarray*}
The formula $\Delta r^\alpha= \alpha r^{\alpha-1} \Delta r
+\alpha(\alpha-1)r^{\alpha-2}\|\grad r\|^2$ applied for $r=r_X$ yields
\begin{equation*}
        {r_X}^{2-\alpha} \Delta_g {r_X}^{\alpha} = \alpha r_X \Delta_g
        r_X+ \alpha(\alpha-1) \|\grad_g r_X\|_g^2.
\end{equation*}
We apply this for  $\alpha=-(n-2)/2$ and $\alpha=2$ and obtain
\begin{equation*}
        {r_X}^{\frac{n+2}2} \Delta_g {r_X}^{-\frac{n-2}2}=
        -\frac{n-2}4 \Delta_g r_X^2 + \frac{n^2-4}4 \|\grad_g
        r_X\|_g^2.
\end{equation*}
{}From the Gauss lemma applied to $\bar g$ it follows that $r_X^2\in
C^\infty(M)$. In Proposition~\ref{prop.these.are.lifts} we have shown
that $\|\grad_g r_X\|_g^2\in \CI([M:X])$, thus 
\begin{equation*}
        {r_X}^{\frac{n+2}2} \Delta_g
        {r_X}^{-\frac{n-2}2}\in \CI([M:X]).
\end{equation*}
Using then $\grad_g r_X^2\in \maW$, also proven in
Proposition~\ref{prop.these.are.lifts}, the lemma follows.
\end{proof}

We shall need the following result as well.

\begin{lemma}\label{lemma.distance}
Using the notation of Lemma \ref{lemma.metric1}, let $X \subset
Y \subset M$ be submanifolds with corners. Let $d_g$ (respectively,
$d_h$) be a smoothed distance function to $Y$ in the metric $g$
(respectively, in the metric $h= r_X^{-2}g$). Then the quotient
$r_X^{-1}d_g/d_h$, defined on $M \smallsetminus (Y \cup \pa M)$,
extends to a smooth function on $[M: X]$.
\end{lemma}

\begin{proof} This is a local statement, so it can be proved
using local coordinates. See \cite{bmnz} for a similar result.
\end{proof}

\subsection{Iterated Blow-ups of Lie-manifolds}
We now iterate the above constructions to blow up a clean
family of submanifolds.

Let us fix for the remainder of this section the following notation:
$(M,\maV)$ is a fixed Lie manifold and $\maX$ is a fixed clean 
family of submanifolds with corners. As discussed at the
end of Section~\ref{sec.diff.struc}, we can assume that
$\maX=(X_i\,|\,i=1,2,\ldots,k)$ is admissibly ordered.  We denote by
$P=[M: \maX]$ the blow-up of $M$ with respect to $\maX$ and by
$\beta : P \to M$ the blow-down map.  Again let
$Y^{(1)},Y^{(2)},\ldots, Y^{(k)}$ and $M^{(0)}, M^{(1)}, \ldots,
M^{(k)}$ be the canonical sequences
associated to $M$ and the admissibly ordered family $\maX$, 
see Section~\ref{sec.diff.struc}, Definition \ref{def.canonical}. Let
$r_\ell : M^{(\ell-1)} \to [0, \infty)$ be a smoothed distance
function to $Y^{(\ell)}$, $1 \le \ell \le k$ in a true metric on
$M^{(\ell-1)}$ (in particular smooth up to the boundary).
Then we denote
\begin{equation}\label{eq.def.rho}
        \rho : = r_1 r_2 \ldots r_k,
\end{equation} 
where the product is first defined away from the singularity, and then
it is extended to be zero on the singular set. Let us notice that
$r_j$ is a defining function for the face corresponding to $Y^{(j)}$
in the blow-up manifold $M$.

We also denote by $r_{\maX}(x)$ the distance from $x$ to $\bigcup \maX
:= \bigcup_{i=1}^k X_i$, again in a true metric.  Let us note for
further use the following simple fact.

\begin{lemma}\label{lemma.r.rho} Using the notation just introduced,
we have that the quotient $r_{\maX}/\rho$, defined first on
$M \smallsetminus (\bigcup \maX )$, extends to a continuous, nowhere
zero function on $P$. In particular, there exists a constant $C>0$
such that
\begin{equation*}
        C^{-1} \rho \le r_{\maX} \le C\rho.
\end{equation*}
\end{lemma}

\begin{proof} 
This follows by induction from Lemma \ref{lemma.distance}, as
in \cite{bmnz}.
\end{proof}

We now show that we can blow up Lie manifolds with respect to a
clean family to obtain again a Lie manifold. Recall that
the blow-down map $\beta : P\to M$ was introduced in
Equation~\eqref{eq.it.blow-down} as the composition $ \beta
:= \beta_{1} \circ \beta_{2} \circ \ldots \circ \beta_{k} : P =
M^{(k)}= [M: \maX] \to M = M^{(0)}.$

\begin{prop}\label{prop.b.trans.fam}
Using the above notation, we have that
\begin{equation*}
        \maW_0
        := \{W \in \Gamma_\beta(TP),\ \beta_*(W\vert_{M\smallsetminus
        \bigcup\maX}) \in \rho (\maV\vert_{M\smallsetminus \bigcup\maX}) \}
\end{equation*}
is isomorphic to $\maV$ as a $\CI(M)$-module. Let
\begin{equation*} 
        \maW := \{f W,\ W \in \maW_0, f \in \CI(P) \}.
\end{equation*} 
Then $\maW$ is a Lie algebra isomorphic to $\CI(P)
\otimes_{\CI(M)} \maV$ as a $\CI(P)$-module and hence $\maW$ is
a finitely generated, projective module over $\CI(P)$, and $(P,\maW)$
is a Lie manifold, which is
isomorphic to the Lie manifold obtained by iteratively blowing up the
Lie manifold $(M, \maV)$ along the submanifolds $Y^{(\ell)}$, $1 \le \ell \le k$.
\end{prop}

\begin{proof} 
Again, this follows by induction from Lemmas
\ref{lemma.distance}, \ref{lemma.r.rho}, and Theorem \ref{theorem.Lie.P}.
\end{proof}

The Lie manifold $(P, \maW) = ([M: \maX], \maW)$ is called the {\em
blow-up of the Lie manifold $(M, \maV)$ along the clean
family $\maX$.}

\begin{prop}\label{prop.Laplace}
Using the notation of the Proposition \ref{prop.b.trans.fam}, let
$A \to M$ be the Lie algebroid associated to $\maV$, so that
$\maV \simeq \Gamma(A)$.  Let us choose a metric $g$ on $A$. Let $B$
be the Lie algebroid associated to $(P, \maW)$. Then the restriction
of $\rho^{-2} g$ to $M \smallsetminus (\bigcup \maX \cup \pa M)$
extends to a smooth metric $h$ on $B$. Let $\Delta_g$ and $\Delta_h$
be the associated Laplace operators.  Then
\begin{equation*}
        \rho^2 \Delta_g - \Delta_h \in \Diff^1_{\maW}(P). 
\end{equation*}
In particular, $\rho^2 \Delta_g$ is elliptic in $\Diff^2_{\maW}(P)$.
\end{prop}

\begin{proof}
This proposition follows from Lemma~\ref{lemma.metric1} by induction.
\end{proof}

We complete this section with a description of the Sobolev space of the
blow-up.

\begin{prop}\label{proposition.sobolev}
Using the notation of Lemma \ref{lemma.r.rho} and of
Proposition \ref{prop.b.trans.fam}, we have
\begin{equation*}
        W^{k,p}(P, \maW) := \{u: M \to \CC, \rho^{j} V_1 \ldots V_j
        u \in L^p(M, d\vol_g),\, \forall V_1, \ldots, V_j \in \maV, \,
        j \le k\,\}\, .
\end{equation*}
\end{prop}

\begin{proof} 
This follows from Lemmas \ref{lemma.sobolev} and \ref{lemma.r.rho}.
\end{proof}

\section{Regularity of eigenfunctions}\label{sec.reg.eigen}

We now provide the main application of the theory developed in the
previous sections.

\subsection{Regularity of multi-electron eigenfunctions} 
Let us consider $\RR^{3N}$ with the standard Euclidean metric.  We
radially compactify $\RR^{3N}$ as follows.  Using the diffeomorphism
$\phi:\RR^{3N}\to B_1(0)$, $x\mapsto \frac{2 \arctan |x|}{\pi|x|} x$
we view $\RR^{3N}$ as the open standard ball $\RR^{3N}$.  The
compactification $M = \overline{\RR^{3N}}$ is then a manifold with
boundary together with a diffeomorphisms from $M$ to the closed
standard ball, extending $\phi$.  The compactification $M$ carries a
Lie structure at infinity~$\maV_{sc}$  
\cite{aln1, SchroheSG, LauterNistor, Parenti, MelroseScattering} which 
is defined as follows.  Let $r_{\infty}$ be a defining function of the
boundary of $M = \overline{\RR^{3N}}$, for example, we can take
$r_{\infty}(x) = (1 + |x|^2)^{-1/2}$. We extend $x_1:=r_\infty$
locally to coordinates $x_1,x_2,\ldots,x_N$, defined on a neighborhood
of a boundary point. In particular $x_2, \ldots, x_N$ are coordinates
of the boundary. In these coordinates $\maV_{sc}$ is generated by
$r_{\infty}^2\pa_{r_\infty}, r_{\infty}\pa_{x_j}$, $j=2, \ldots , N$.
Thus $\maV_{sc} = r_{\infty}\maV_M$, with $\maV_M$ defined in
Equation~\eqref{eq.maV_M}.  We can then choose the metric on
$\maV_{sc}$ so that the induced metric on $M_0$, the interior of $M$,
is the usual Euclidean metric on $\RR^{3N}$.

Motivated by the specific form of the potential $V$ introduced in
Equation~\eqref{eq.def.V}, let us now introduce the following family
of submanifolds of $M = \overline{\RR^{3N}}$.  Let $X_{j}$ be the
closure in $M$ of the set $\{x = (x_1, \ldots, x_N), x_j =
0 \in \RR^3\}$. Let us define similarly $X_{ij}$ to be the closure in
$M$ of the set $\{x = (x_1, \ldots, x_N), x_i =
x_j \in \RR^3\}$. Let~$\maS$ be the family of consisting of all
manifolds $X_{j}, X_{ij}$ for which the parameter functions $b_j$ and
$c_{ij}$ are non-zero, together with their intersections. The family
$\maS$ will be called the {\em multi-electron} family of singular
manifolds.

\begin{prop} \label{prop.S.transverse} 
The multi-electron family of singular manifolds $\maS$ is a clean family.
\end{prop}

\begin{proof}
Let $\maY = \{Y_j\}$ be the family of all finite intersections of the
sets $X_j$. We need to prove that $T_x(\bigcap Y_{j_k}) = \bigcap T_x
Y_{j_k}$.  At a point $x \in \RR^{3N}$ this is obvious, since each
$Y_j$ is (the closure of) a linear subspace close to $x$. For $x$ on
the boundary of $M$, we notice that $\maY$ has a product structure in
a tubular neighborhood of the boundary of $M$.
\end{proof}

Let $(\SS, \maW) := ([M: \maS], \maW)$ be the blow-up of the Lie
manifold $(M=\overline{\RR^{3N}}, \maV_{sc})$, given by
Proposition \ref{prop.b.trans.fam}, and $\rho$ be the function
introduced in \eqref{eq.def.rho}. Note that the definition of~$\SS$
and $\maW$ depend on which of the $b_j$ and $c_{ij}$ are allowed to be
non-zero. Let $V$ be the potential considered in Equation~\eqref{eq.def.V}:
\begin{equation*} 
  V(x) = \sum_{1 \le j \le N} \frac{b_j}{|x_j|} 
  + \sum_{1 \le i < j \le N} \frac{c_{ij}}{|x_i-x_j|},
\end{equation*}
where $x = (x_1, x_2, \ldots, x_N) \in \RR^{3N}$, $x_j \in \RR^3$.  We
allow $b_j, c_{ij} \in \CI(\SS)$, which is important for some
applications to the Hartree--Fock and Density Functional
Theory. We endow $\SS$ with the volume form defined by a
compatible metric and we then define $L^p(\SS)$ accordingly.

\begin{thm}\label{thm.regularity} 
The blow-up $(\SS, \maW)$ of the scattering manifold
$(M=\overline{\RR^{3N}}, \maV_{sc})$ has the following properties:
\begin{enumerate}[(i)]
\item $\rho V \in r_{\infty} \CI(\SS)$.
\item  $\rho^2(-\Delta + V) \in \Diff_{\maW}(\SS)$ and
is elliptic in that algebra.  
\item Let  $x_H$ be a defining function of the face $H$ and $a_H \in \RR$,
for each hyperface $H$ of~$\SS$. Denote $\chi = \prod_H x_H^{a_H}$
and assume that $u \in \chi L^p(\SS)$ satisfies $(-\Delta + V)u
= \lambda u$, $1<p<\infty$, for some $\lambda \in \RR$. Then
$u \in \chi W^{m,p}(\SS, \maW)$ for all $m \in \ZZ_+$.
\end{enumerate}
\end{thm}

\begin{proof}
(i) Let us choose the compatible metric to be the Euclidean metric and
choose the boundary defining function $r_{\infty}$ for the boundary
(sphere) at infinity of $\RR^{3N}$ to satisfy $r_{\infty}(x) = 1/|x|$
for $|x|$ large.  Let $X$ be any of the manifolds $X_j
:= \overline{\{x_j = 0\}}$ or $X_{ij} := \overline{\{x_i = x_j\}}$
defining $\maS$ (the closures are all in $M$). We shall denote by
$r_X$ the distance to $X$ in a true (bounded) metric on $M$ and by
$d_X$ the distance to $X$ in the Euclidean metric. For example, if
$X \cap \RR^{3N} = X_j\cap \RR^{3N} = \{x_j = 0 \in \RR^{3} \}$, then
$d_X(x) = |x_j|$. We claim that the function $\phi :=
r_{\infty}d_{X}/r_X$ extends to a smooth and positive function on
$[M:X]$. We will assume that the bounded metric is a product metric
near the boundary in the standard (polar coordinates) tubular
neighborhood $U = S^{3N-1} \times [0, \epsilon)$ of $S^{3N-1}$. We can
also assume $r_{\infty}(x', t) = t$. In the interior of $M$, the
smoothness and positivity of $\phi$ follows from the fact that if $V$
is a linear subspace of $\RR^{3N}$ and if $g_1$ and $g_2$ are two
scalar products on $\RR^{3N}$ with associated distance functions $d_1$
and $d_2$, then $x\mapsto d_1(x,V)/d_2(x;V)$ extends to a smooth and
positive function on $[\RR^{3N}:V]$. At the boundary of $M$ the
argument uses also homogeneity. Both functions $r_X(x', t)$ and
$r_{\infty}(x', t) d_X(x', t)$ are in fact independent of $t \in
[0, \epsilon)$. Therefore $\phi(x',t)$ is independent of $t$ for $t$
small. Since the function $\phi$ was proved to be smooth for $t > 0$,
the claim follows.

It follows that $\phi$ is a smooth function also on $\SS=[M: \maS]$,
because $\CI([M:X]) \subset \CI([M:\maS])$. Moreover, $\phi$ is
nowhere zero, so we also have $\phi^{-1} \in \CI(\SS)$. Since $V$ is a
sum of terms of the form $d_{X}^{-1}$, it is enough to show that
$\rho/d_{X} \in r_{\infty} \CI(\SS)$. But $\rho = \psi r_X$ for some
smooth function $\psi \in \CI(\SS)$ and hence
\begin{equation*}
        \rho/d_{X} = \psi r_X/d_X
        = \psi \phi^{-1} r_{\infty} \in r_{\infty} \CI(\SS).
\end{equation*}

(ii) follows from Propositions \ref{prop.Laplace}
and \ref{prop.S.transverse} using also (i) just proved.

(iii) is a direct consequence of the regularity result in \cite{ain},
Theorem \ref{thm.ain}, because $\rho^{2}(-\Delta + V - \lambda)$ is
elliptic, by (ii). The proof is now complete.
\end{proof}

Note that it follows from Proposition \ref{proposition.sobolev} and
the definition of $\maV_{sc}$ that
\begin{equation}\label{eq.sobolev.SS}
        W^{k,p}(\SS, \maW) := \{u: \RR^{3N} \to \CC, \rho^{|\alpha| + 3N/2}
        \pa^\alpha u \in L^p(\RR^{3N}),\, \
        |\alpha| \le k\,\}\, .
\end{equation}

We are now ready to prove our main result, as stated in Equation~\eqref{eq.main}.

\begin{thm}\label{theorem.main} 
Assume $u \in L^2(\RR^{3N})$ is an eigenfunction of $\cH := -\Delta +
V$, then 
\begin{equation*}
        u \in \maK_{a}^{m}(\RR^{3N},r_S) = \rho^{a
        -3N/2}W^{m,2}(\SS, \maW)
\end{equation*}
for all $m\in \ZZ_+$ and for all $a \le 0$.
\end{thm}

\begin{proof} We have that $L^2(\RR^{3N}) = \rho^{-3N/2} L^2(\SS)$ 
since the metric on $\SS$ is $g_\SS = \rho^{-2} g_{\RR^{3N}}$. The
function $\rho$ is a product of defining functions of faces at
infinity, so $\rho^{-3N/2} = \chi$, for some $\chi$ as in
Theorem \ref{thm.regularity} (iii). The result then follows from
Theorem \ref{thm.regularity} (iii).
\end{proof}

\subsection{Regularity in the case of one electron and several heavy nuclei.}
Let us now consider $S = \{P_1, P_2, \ldots P_m\} \in \RR^{3}$, let
$M$ be the scattering calculus Lie manifold obtained by radially
compactifying $\RR^3$, as in the previous subsection. So $N=1$ in this
section, but we allow several fixed nuclei. Let us blow it up with
respect to the set $S$, obtaining a manifold with boundary $\SS$. Let
$\maW$ be the structural Lie algebra of vector fields on $\SS$
obtained blowing up the scattering calculus on $M$.

Let $V_0, k_j : \SS \to \RR$ be smooth functions, $j=1, 2, 3$. Let
$r_S : \SS \to [0, 1]$ be a smooth function that is equal to $0$ on
the faces corresponding to the singular points in $S$ and equal to $1$
in a neighborhood of the hyperface coming from the ball
compactification of $\RR^3$, \ie the face at infinity before the
blowup.  We assume that $dr_S \neq 0$ on the faces corresponding to
the set of singular points $S$.  As $S$ is a compact set, we can
assume in this subsection that $r_S(x)$ is the euclidian distance from
$x$ to $S$ if $x \in \RR^{3} \smallsetminus S$ is close to $S$. We
have $r_S = \rho$ in the notation of the previous subsection.

In view of further applications to operators that arise in the study
of periodic potentials, in this subsection we shall consider
eigenfunctions of the operator
\begin{equation}\label{eq.def.m.hamilt}
        \cH_m = - \sum_{j=0}^3 (\pa_j - i k_j)^2 +
        V_0/r_S,
\end{equation}
which is the {\em magnetic} version of the Schr\"{o}dinger
operator \eqref{hamilt}. For possible applications to the periodic
case, the case where $k_j$ are constants is the most important case,
but our results are more general.  Recall that the spaces $H^m(\SS)$
were introduced in Equation \eqref{eq.def.Hm}.  Also, let us notice
that $e^{-\epsilon |x|}$ is a smooth function on $\SS$, so
multiplication by this function maps the spaces $H^m(\SS)$ to
themselves.

\begin{thm}\label{thm.regularity2} 
Let $u \in L^2(\RR^3)$ be such that $\cH_m u = \lambda u$, in
distribution sense. Then
\begin{enumerate}[(i)]
\item\label{lab1} $r_S^2 e^{\mu |x|} \cH_m e^{-\mu |x|} \in \Diff_{\maW}(\SS)$, $\mu \in \RR$,
is elliptic.
\item\label{lab2} $u \in r_S^{-3/2} H^{m}(\SS) = \maK_{0}^m(\RR^3,r_S)$ for all $m$. 
\item\label{lab3} If $-\lambda > \epsilon > 0$, then
$u \in r_S^{-3/2} e^{-\epsilon |x|} H^{m}(\SS)$ for all $m$.
\end{enumerate}
\end{thm}

\begin{proof} 
The first part, \eqref{lab1}, is a direct calculation, completely
similar to Theorem \ref{thm.regularity}. To prove $(ii)$, we notice
that $L^2(\RR^{3}) = r_S^{-3/2}H^0(\SS)$. Then \eqref{lab2} is an
immediate consequence of the regularity theorem of \cite{ain}.
Finally, we have that $v = e^{\epsilon |x|} u \in L^2(\RR^3) =
r_S^{-3/2}H^{0}(\SS)$ by \cite{AgmonDecay}, since $-\lambda > \epsilon
> 0$. It is also an eigenfunction of $H_1 := e^{\epsilon |x|} \cH_{m}
e^{-\epsilon |x|}$. The result of~\eqref{lab3} then follows from the
ellipticity of~$r_S^{2}H_1$, by~\eqref{lab1}, and by the regularity
theorem of~\cite{ain}, Theorem~\ref{thm.ain}.
\end{proof}

See also \cite{CarmonaSimon1, Cornean, GeorgescuIftimovici,
PuriceDecay, VasyExp} and the references therein for more on the decay
of eigenfunctions. 
See also \cite{Kato51, KatoBook} for additional general properties of
the Hamiltonian operators arising in Quantum Mechnics.

To get an improved regularity in the index $a$, we shall need the
following result of independent interest. Let us replace $\RR^{3}$ by
$\RR^{N}$ in the following result, while keeping the rest of the
notation unchanged. In particular, $S \subset \RR^{N}$ is a finite
subset and $r_S(x) \in [0, 1]$ is the distance from $x$ to $S$ for $x$
close to $S$ and is equal to $1$ in a neighborhood of the hyperface at
infinity before the blow-up of the singular points.


As usual we define $\maK_{-a}^{-m}(\RR^N,r_S)$ to be the dual 
of $\maK_{a}^{m}(\RR^N,r_S)$ with respect to the pairing 
$(f_1,f_2):=\int_{\RR^N} f_1 f_2$, where 
$\maK_{a}^{m}(\RR^N,r_S)$ was defined in \eqref{eq.def.sobw}. 

\begin{thm}\label{thm.Delta}  Let $|a| < (N-2)/2$, then
\begin{equation*}
        \Delta - \mu
        : \maK_{a+1}^{m+1}(\RR^N,r_S) \to \maK_{a-1}^{m-1}(\RR^N,r_S)
\end{equation*}
is an isomorphism for $\mu>0$ large enough.
\end{thm}

\begin{proof} 
We begin by recalling the classical Hardy's inequality, valid for $u \in
H^1(\RR^N)$:
\begin{equation}\label{eq.Hardy}
        c_N^2 \int_{\RR^N} \frac{|u|^2}{|x|^2}
        dx \leq \int_{\RR^N} |\nabla u|^2 dx,
\end{equation}
with $c_N = (N-2)/2$ (see for example \cite{Zuazua} and the
references therein). A partition of unity argument then implies that
for any $\delta>0$ there exists $\mu = \mu_\delta >0$ such that
\begin{equation}\label{eq.Hardy2}
        (1-\delta) c_N^2 \int_{\RR^N} |r_S^{-1} u|^2
        dx \leq \int_{\RR^N} \big( |\nabla u|^2 + \mu|u|^2 \big) dx.
\end{equation}

We can assume that $|\nabla r_S| \le 1$. Let us assume
$u \in \CIc(\RR^{3} \smallsetminus S)$, which is a dense subset of
$\maK_{a}^m(\RR^{N}, r_S)$ for all $m$ and $a$, by \cite{ain}.  Let $|a|
<(N-2)/2$. We shall denote $(u, v) = \int_{\RR^{N}} uv\, dx$, as
usual. Let us regard $r^{a}$ and $r^{-a}$ as multiplication
operators. Let us now multiply Equation \eqref{eq.Hardy2} with
$1-\delta$ and use $\nabla(r_S^{a}u) = a r_S^{a-1} u \nabla r_S +
r_S^{a} \nabla u$ to obtain
\begin{eqnarray*}
        \big( (\mu - r_S^{-a}\Delta r_{S}^{a}) u, u \big) &=& \mu(u,
      u) + (\nabla r_{S}^{a} u, \nabla r_{S}^{-a} u ) \\
        & = & \mu(u, u) + (r_{S}^{a} \nabla u, r_{S}^{-a} \nabla u ) +
        a( r_{S}^{-1}(\nabla r_{S}) u, \nabla u) \\
        & & - a( \nabla u, r_{S}^{-1}(\nabla r_{S}) u) - a^2(r_{S}^{-1}
        (\nabla r_{S}) u , r_{S}^{-1} (\nabla r_{S}) u )\\
        & \ge &\mu(u, u) + (\nabla u, \nabla u) - a^2(r_{S}^{-1} u ,
        r_{S}^{-1} u)\\
        & \ge &( (1-\delta)^2 c_N^2 - a^2)(r_{S}^{-1} u , r_{S}^{-1} u)
        + \delta (\nabla u, \nabla u) \\
        &\ge &\delta \|u\|_{\maK_1^1}^2.
\end{eqnarray*}
For $\delta >0$ small enough ($(1-\delta)^2 c_N^2 - \delta \ge a^2$).
This means that the continuous map
\begin{equation*}
        P_{a,\mu} := \mu - r_{S}^{-a}\Delta r_{S}^{a}
        : \maK^1_{1}(\RR^{N},r_S) \to \maK^{-1}_{-1}(\RR^{N},r_S)
\end{equation*}
satisfies 
\begin{equation*}
        \|P_{a,\mu} u \|_{\maK_{-1}^{-1}} \|u \|_{\maK_{1}^{1}} \ge
        (P_{a,\mu} u, u) \ge \delta \|u\|_{\maK_1^1}^2,
\end{equation*}
and hence $\|P_{a,\mu} u \|_{\maK^{-1}_{-1}(\RR^{N})}
\ge \delta \|u\|_{\maK^1_{1}(\RR^{N})}$, for $\mu>0$ 
large and some $\delta > 0$.  It follows that $P_{a,\mu}$ is injective
with closed range for all $|a|<(N-2)/2$. Since the adjoint of
$P_{a,\mu}$ is $P_{-a,\mu}$, it follows that $P_{a,\mu}$ is also
surjective, and hence an isomorphism by the Open Mapping Theorem. The
regularity result of \cite{ain} (Theorem \ref{thm.ain}) shows that
$P_{a,\mu} := \mu - r_{S}^{-a} \Delta r_{S}^{a}
: \maK^{m+1}_{1}(\RR^{N},r_S) \to \maK^{m-1}_{-1}(\RR^{N},r_S)$ is also an
isomorphism for all $m$. The result follows from the fact that $r_S^b
: \maK^{m}_{c} (\RR^{N},r_S) \to \maK^{m}_{c+b} (\RR^{N},r_S)$ is an
isomorphism for all $b$, $c$, and $m$ \cite{ammann.nistor:07}.
\end{proof}

We are ready to prove the main result of this subsection.

\begin{thm}\label{thm.regularity3} 
Let $u \in L^2(\RR^3)$ be such that $\cH_m u = \lambda u$, in
distribution sense. Then $u \in \maK_{a}^{m}(\RR^{3},r_S) = r_S^{a-3/2}
H^{m}(\SS)$ for all $m \in \ZZ_+$ and all $a<3/2$.
\end{thm}

\begin{proof}
Let us first notice that the operator $Q := \cH_m +\Delta$ is a
bounded operator $\maK_{a}^{m}(\RR^{3},r_S) \to \maK_{a-1}^{m-1}(\RR^{3},r_S)$
for all $a$ and $m$.  Assume that $u \in L^2(\RR^{3})$ satisfies
$-\cH_m u = \lambda u$. Then we know that
$u \in \maK_{0}^{m}(\RR^{3},r_S)$ for all $m$ by
Theorem~\ref{thm.regularity2}. Hence
\begin{equation*}
        f := (\Delta - C) u = Qu + (\lambda-C)
        u \in \maK_{-1}^{m-1}(\RR^3,r_S).
\end{equation*}
For large $C$ we can invert $\Delta - C$, and thus we obtain 
$u = (\Delta - C)^{-1}f \in \maK_{1}^{m+1}(\RR^3,r_S) = (\Delta -
C)^{-1}\maK_{-1}^{m-1}(\RR^3,r_S)$, by Theorem~\ref{thm.Delta}. But then
$f = Qu + (\lambda-C) u \in \maK_{0}^{m}(\RR^3,r_S) \subset \maK_{-1 +
a}^{m-1}(\RR^3,r_S)$ for any $a<1/2$. We can then repeat this argument to
obtain $u = (\Delta - C)^{-1}f \in \maK_{1+a}^{m+1}(\RR^3,r_S)$ for any
$a < 1/2$ and any $m$, as claimed.
\end{proof}

See \cite{Ferrero, Flad1, HunsickerNistorSofo} for an approach to the
singularities of one electron Hamiltonians using the theory of
singular functions for problems with conical singularities. The
regularity at the origin in the above theorem is, in fact, a simple
consequence of the theory of singular functions. For $V_0$ real
analytic and $k_j=0$, the regularity at the origin is also an
immediate consequence of the analytic regularity result proved
in \cite{Fournais3}.

It would be interesting to extend our results in the case of magnetic
fields \cite{Georgescu, HunsickerNistorSofo, PuriceMagnetic, Siedentop}.
In addition to the above extensions, one would have to look into the
issues that arise in the numerical approximation of solutions of
partial differential equations in spaces of high dimension (the so
called ``curse of dimensionality''). Let us mention in this regard the
papers \cite{Griebel, Hamaekers, Schwab} and the references therein,
where the issue of approximation in high dimension is discussed.

\onlyarxive{
\appendix 

\section{$b$-tangent bundle and partial $b$-structure on $[M:X]$}\label{subsec.btg}\label{app.ex}

In this example we give an example for a Lie manifold as explained
in Subsection~\ref{subsec.liemanifolds}. The content of this
section was removed in the printed version, as we were asked to 
shorten the article.
 
Important examples of Lie manifolds are Melrose's $b$-manifolds.  Let
$N$ be a manifold with corners. The $b$-tangent bundle is a Lie
algebroid $T^bN$ with an anchor map $\rho:T^bN\to TN$ such that $\rho$
induces a $\CI(M)$-module isomorphism, and
$\Gamma(T^bN)\cong \maV_N$. Recall that $\maV_N$ was defined in
Equation~\eqref{eq.maV_M}. The Lie algebroid $T^bN$ is hereby
determined up to isomorphisms of Lie-algebroids.
 
Now we assume that, following \cite{ain}, the boundary hyperfaces
$\{H_1,\ldots, H_k\}$ of $N$ are divided into two sets
$\maT=\{H_1,\ldots,H_r\}$ (the so-called {\it true} boundary faces)
and $\maF=\{H_{r+1},\ldots, H_k\}$, (the so-called boundary faces {\it
at infinity}).  The cases $r=0$ and $r=k$ are not excluded, \ie one of
these sets might be empty. Then one carries out the
$b$-con\-struc\-tion only at the boundary faces at infinity. In other
words, one defines $T^{b\maF}N$ as a vector bundle with anchor map
inducing an isomorphism between $\Gamma(T^{b\maF}N)$ and the set
$\maV^\maF_N$ of vector fields, tangent to the boundaries at infinity.
As above $T^{b\maF}N$ is hereby determined up to isomorphism of
Lie-algebroids.

This bundle plays an important role on $N=[M:X]$ where $X$ is a
submanifold with corners of the manifold with corners $M$.  The
boundary hyperfaces of $[M:X]$ arising from boundary hyperfaces of $M$
are considered as true boundary, whereas the boundary faces obtained
from the blow-up around $X$, are considered as boundary at infinity.
In this situation $T^{b\maF}N$ will be denoted as $T^{bX}[M:X]$.
}


\end{document}